\documentclass[11pt]{article}

\date{\today} 

\title{Survivable Network Design for Group Connectivity in Low-Treewidth Graphs}

\author{
Parinya Chalermsook\thanks{
   Aalto University, Finland. 
   {\bf email:} \texttt{chalermsook@gmail.com}}
\and
Syamantak Das\thanks{
   Indraprastha Institute of Information Technology Delhi, India.
   {\bf email:} \texttt{syamantak@iiitd.ac.in}}
\and
Guy Even\thanks{
   Tel-Aviv University, Israel. 
   {\bf email}: \texttt{guy@eng.tau.ac.il}}
\and
Bundit Laekhanukit\thanks{
   Max-Planck-Institut f\"ur Informatik, Germany \& 
   Shanghai University of Finance and Economics, China.
   {\bf email}: \texttt{blaekhan@mpi-inf.mpg.de}}
\and
Daniel Vaz\thanks{
   Max-Planck-Institut f\"ur Informatik, Germany \& 
   Graduate School of Computer Science, Saarland University, Germany.
   {\bf email}: \texttt{ramosvaz@mpi-inf.mpg.de}}
}

\usepackage{fullpage}
\usepackage{amsfonts,amssymb,amsthm,amsmath}
\usepackage{thmtools}
\usepackage{mathpazo}
\usepackage{hyperref}
\usepackage{cleveref}
\PassOptionsToPackage{ocgcolorlinks}{hyperref}

\declaretheorem[numberwithin=section]{theorem}
\declaretheorem[numberlike=theorem]{lemma}

\declaretheorem[numberlike=theorem]{observation}

\newcommand{\shortparagraph}[1]{\paragraph{#1}}

\usepackage{amsfonts,amssymb,amsthm,amsmath}
\usepackage{amstext}

\usepackage{xspace}
\usepackage{graphicx}
\usepackage{colordvi}
\usepackage{color}
\usepackage{paralist}
\usepackage{boxedminipage}
\usepackage[percent]{overpic}
\usepackage{wrapfig}
\usepackage{lmodern}

\usepackage{hyperref}
\usepackage{cleveref}

\crefname{theorem}{Theorem}{Theorems}
\crefname{corollary}{Corollary}{Corollaries}
\crefname{lemma}{Lemma}{Lemmas}
\crefname{proposition}{Proposition}{Propositions}
\crefname{claim}{Claim}{Claims}
\crefname{definition}{Definition}{Definitions}
\crefname{observation}{Observation}{Observations}

\newcommand{\NP}{\mbox{\sf NP}}

\newcommand{\polylog}[1]{\mathrm{polylog(#1)}}
\newcommand{\DTIME}{\mbox{\sf DTIME}}

\newcommand{\ts}{t^{\star}}

\newcommand{\vset}{{\mathcal V}}
\newcommand{\cons}[1]{\stackrel{#1}{\longleftrightarrow}}

\renewcommand{\phi}{\varphi}

\newcommand{\poly}{\operatorname{poly}}

\newcommand{\RootedSNDP}{\textsf{Rooted SNDP}\xspace}
\newcommand{\RootedECSNDP}{\textsf{Rooted EC-SNDP}\xspace}
\newcommand{\RootedVCSNDP}{\textsf{Rooted VC-SNDP}\xspace}
\newcommand{\RestrictedGroupSNDP}{\textsf{Restricted Group SNDP}\xspace}
\newcommand{\RootedGroupSNDP}{\textsf{Rooted Group SNDP}\xspace}
\newcommand{\SubsetkEC}{\textsf{Subset $k$-EC}\xspace}
\newcommand{\SubsetkVC}{\textsf{Subset $k$-VC}\xspace}

\newcommand{\gst}{\mbox{\sf GST}\xspace}
\newcommand{\sndp}{\mbox{\sf SNDP}\xspace}
\newcommand{\ecsndp}{\mbox{\sf EC-SNDP}\xspace}
\newcommand{\vcsndp}{\mbox{\sf VC-SNDP}\xspace}

\newcommand{\pset}{{\mathcal P}}

\newcommand{\cset}{{\mathcal C}} 
\newcommand{\sset}{{\mathcal S}} 
\newcommand{\tset}{{\mathcal T}} 
\newcommand{\xset}{{\mathcal X}} 
\newcommand{\tc}[0]{\ensuremath{\operatorname{tc}}}
\newcommand{\restr}[1]{\ensuremath{\bigl.#1\bigr|}}

\newcommand{\ttil}[0]{\ensuremath{\tilde t}}
\newcommand{\rootn}[0]{\ensuremath{\operatorname{root}}}

\begin{document}
\maketitle

\begin{abstract}
{\small
In the {\em Group Steiner Tree} problem (GST), 
we are given a (vertex or edge)-weighted graph $G=(V,E)$ on $n$ vertices, 
together with a root vertex $r$ and a collection of 
groups $\{S_i\}_{i \in [h]}: S_i \subseteq V(G)$. 
The goal is to find a minimum-cost subgraph $H$ that 
connects the root to every group.
We consider a fault-tolerant variant of GST, 
which we call \textsf{Restricted (Rooted) Group SNDP}.
In this setting, 
each group $S_i$ has a demand $k_i \in [k], k\in {\mathbb N}$,
and we wish to find a minimum-cost subgraph $H\subseteq G$ such that, 
for each group $S_i$, there is a vertex in the group that is 
connected to the root via $k_i$ (vertex or edge) disjoint paths. 
   
While GST admits $O(\log^2 n \log h)$ approximation,
its higher connectivity variants are known to be {\em Label-Cover} hard, 
and for the vertex-weighted version, the hardness holds even when $k=2$
(it is widely believed that there is no subpolynomial approximation
for the Label-Cover problem [Bellare et al., STOC 1993]).
More precisely, the problem admits no $2^{\log^{1-\epsilon}n}$-approximation
unless $\NP\subseteq\DTIME(n^{\polylog{n}})$.
Previously, positive results were known only for the edge-weighted version
when $k=2$ 
[Gupta et al., SODA 2010; Khandekar et al., Theor. Comput. Sci., 2012]
and for a relaxed variant where $k_i$ disjoint paths from $r$ may end
at different vertices in a group [Chalermsook et al., SODA 2015], 
for which the authors gave a bicriteria approximation.
For $k \geq 3$, there is no non-trivial approximation algorithm known
for edge-weighted \RestrictedGroupSNDP, except for 
the special case of the relaxed variant on trees (folklore).

Our main result is an $O(\log n \log h)$ approximation for 
\RestrictedGroupSNDP that runs in time $n^{f(k, w)}$,
where $w$ is the treewidth of the input graph.
This nearly matches the lower bound when $k$ and $w$ are constants. 
The key to achieving this result is a non-trivial extension of 
a framework introduced in [Chalermsook et al., SODA 2017]. 
This framework first embeds all feasible solutions to the problem into a 
dynamic program (DP) table.
However, finding the optimal solution in the DP table remains intractable.
We formulate a linear program relaxation for the DP and obtain 
an approximate solution via randomized rounding. 
This framework also allows us to systematically construct 
DP tables for high-connectivity problems. 
As a result, we present new exact algorithms for several variants of survivable network design problems in low-treewidth graphs.  
}
\end{abstract} 

\section{Introduction}

Network design is an important subject in computer science and 
combinatorial optimization.
The goal in network design is to build a network that meets 
some prescribed properties while minimizing the construction cost.
{\em Survivable network design problems} (\sndp) are a class of problems 
where we wish to design a network that is resilient against link or node failures. 

These problems have been phrased as optimization problems on graphs, where we are given
an $n$-vertex (undirected or directed) graph $G=(V,E)$ with 
costs on edges or vertices together with a connectivity requirement 
$k: V\times V\rightarrow\mathbb{N}$.
The goal is to find a minimum-cost subgraph $H\subseteq G$, such that 
every pair $u,v\in V$ of vertices are connected by $k(u,v)$ edge-disjoint 
(resp., openly vertex-disjoint) paths.
In other words, we wish to design a network in which every pair of vertices 
remains connected (unless $k(u,v)=0$), even after removing $k(u,v)-1$ edges (or vertices). 
The {\em edge-connectivity} version of \sndp (\ecsndp) models the existence of link failures
and the {\em vertex-connectivity} (\vcsndp) models the existence of both link and node failures.
These two problems were known to be NP-hard and 
have received a lot of attention in the past decades
(see, e.g., \cite{Jain01,ChuzhoyK12,Nutov12,WilliamsonGMV95}).

While \vcsndp and \ecsndp address the questions that arise from designing telecommunication networks,
another direction of research focuses on the questions that arise from 
media broadcasting as in cable television or streaming services.
In this case, we may wish to connect the global server to a single local server in each community,
who will forward the stream to all the clients in the area through their own local network.
The goal here is slightly different from the usual \sndp, as it is not required
to construct a network that spans every client; instead, we simply need to choose a
local server (or representative), which will take care of connecting to other clients in the same group.
This scenario motivates the {\em Group Steiner Tree} problem (\gst) and its fault-tolerant variant,
 the {\em Rooted Group SNDP}.

In \RootedGroupSNDP, we are given a graph $G=(V,E)$
with costs on edges or vertices, a root vertex $r$, and 
a collection of subsets of vertices called {\em groups},
$S_1,\ldots,S_h$,  together with connectivity demands 
$k_1,\ldots,k_h\in [k]$, $k\in\mathbb{N}$.
The goal in this problem is to find a minimum cost subgraph $H\subseteq G$
such that $H$ has $k_i$ edge-disjoint (or openly vertex-disjoint) paths connecting the root
vertex $r$ to some vertex $v_i\in S_i$, for all $i\in[h]$.
In other words, we wish to choose one representative from each group
and find a subgraph of $G$ such that each representative is $k$-edge-(or vertex)-connected
to the root.

When $k=1$, the problem becomes the well-known 
{\em Group Steiner Tree} (\gst) problem. Here, we are given a
graph $G=(V,E)$ with edge or vertex costs, a root $r$
and a collection of subsets of vertices called groups, 
 $S_1,\ldots,S_h\subseteq V$,
and the goal is to find a minimum-cost subgraph $H\subseteq G$ that
has a path to some vertex in each $S_i$, for $i\in[h]$.  
The \gst problem is known to admit an $O(\log^3n)$-approximation algorithm
\cite{GargKR00} and cannot be approximated to a factor of $\log^{2-\epsilon}n$ unless 
$\NP\subseteq\mathrm{ZPTIME}(n^{\polylog n})$ \cite{HalperinK03}. 

The \RootedGroupSNDP generalizes \gst to handle fault tolerance.
 The case where $k=2$ is studied in \cite{KhandekarKN12,GuptaKR10},
culminating in the $\tilde{O}(\log^4n)$-approximation algorithm for the problem.
For $k\geq 3$, there is no known non-trivial approximation algorithm.
It is known among the experts that this
problem is at least as hard as the {\em Label-Cover} problem%
\footnote{The hardness for the case of directed graph was shown in \cite{KhandekarKN12}, 
but it is not hard to show the same result for undirected graphs.}.

Chalermsook, Grandoni and Laekhanukit \cite{ChalermsookGL15} studied a 
relaxed version of the problem in which we are not restricted to connect to a 
single vertex in each group and thus need only $k_i$ edge-disjoint paths
connecting the root vertex to the whole group $S_i$. 
Despite being a relaxed condition, the problem remains as hard as 
the Label-Cover problem, and they only managed to design a bicriteria approximation
algorithm. 

To date, there is no known bicriteria or even sub-exponential-time poly-logarithmic approximation for \RootedGroupSNDP when $k\geq 3$.
The following is an intriguing open question: 
\begin{quote} 
What are the settings (i.e., ranges of $k$ or graph classes) in which \RootedGroupSNDP admits a poly-logarithmic approximation?  
\end{quote} 

In this paper, we focus on developing algorithmic techniques to approach the above question.
We design poly-logarithmic algorithms for a special class of graphs -- graphs with bounded treewidth --
in the hope that it will shed some light towards solving the problem on a more general class of graphs,
for instance, planar graphs (this is the case for the Steiner tree problem,
where a sub-exponential-time algorithm for planar graphs is derived via decomposition into low-tree width instances \cite{PilipczukPSL13}).

Our main technical building block is a dynamic program (DP) that solves rooted versions of \ecsndp and \vcsndp
in bounded-treewidth graphs. However, a straightforward DP computation is not applicable for {\sf Restricted Rooted Group SNDP},
simply because the problem is NP-hard on trees (so it is unlikely to admit a polynomial-size DP-table). 
Hence, we ``embed'' the DP table into a tree  
and devise a polylogarithmic approximation algorithm using randomized rounding of a suitable LP-formulation.
We remark that when the cost is polynomially bounded (e.g., in the Word RAM
model with words of size $O(\log n)$), polynomial-time algorithms
for \ecsndp and \vcsndp follows from {\em Courcelle's Theorem} \cite{Courcelle90,BoriePT92}
 (albeit, with much larger running time).
However, employing the theorem as a black-box 
does not allow us to design approximation algorithms for {\sf Restricted Rooted Group SNDP}.

To avoid confusion between the relaxed and restricted version of \RootedGroupSNDP (usually having the same name in literature), 
we refer to our problem as {\sf Restricted Group SNDP}. (For convenience, we also omit the word "rooted".)

\subsection{Related Work}
\sndp problems on restricted graph classes have also been studied extensively.  
When $k=1$, the problems are relatively well understood.
Efficient algorithms and PTAS have been developed for many graph classes: low-treewidth graphs~\cite{BateniHM11,CyganNPPRW11}, metric-cost graphs~\cite{CheriyanV07}, Euclidean graphs~\cite{BorradaileKM15}, planar graphs~\cite{BorradaileKM09}, and graphs of bounded genus~\cite{BorradaileDT14}.
However, when $k \geq 2$, the complexity of these problems remains wide open. 
Borradaile et al.~\cite{BorradaileDT14,BorradaileZ17} showed an algorithm for $k =1, 2, 3$ on planar graphs, but under the assumption that one can buy multiple copies of edges (which they called {\em relaxed connectivity} setting). 
Without allowing multiplicity, very little is known when $k\geq 2$: Czumaj et al.~\cite{CzumajGSZ04} showed a PTAS for $k=2$ in unweighted planar graphs, and Berger et al.~\cite{BergerG07} showed an exact algorithm running in time $2^{O(w^2)} n$ for the uniform demand case (i.e., $k(u,v) =2$ for all pairs $(u,v)$).  
Thus, without the relaxed assumption, with non-uniform demands or $k > 2$, the complexity of \sndp problems on bounded-treewidth graphs and planar graphs is not adequately understood.  

The technique of formulating an LP from a DP table has been used in literature.
It is known that any (discrete) DP can be formulated as an LP, which is integral \cite{MartinRC90}.
(For Stochastic DP, please see, e.g., \cite{Manne60,d1963,Buyuktahtakin2011,FariasR01}.)
However, the technique of producing a tree structure out of a DP table is quite rare.
Prior to this paper the technique of rounding LP via a tree structure was used
in \cite{GuptaTW13} to approximate the Sparsest-Cut problem.
The latter algorithm is very similar to us.
However, while we embed a graph into a tree via a DP table,
their algorithm works directly on the tree decomposition.
We remark that our technique is based on the previous work
in \cite{ChalermsookDLV17} with almost the same set of authors.

\subsection{Hardness of Approximating Restricted Group SNDP}
As mentioned, it is known among the experts that vertex-cost variant of
the \RestrictedGroupSNDP has a simple reduction for the {\em Label-Cover} problem 
and more generally,
the {\em $k$-Constraint Satisfaction} problem ($k$-CSP).
The original construction was given by Khandekar, Kortsarz and Nutov \cite{KhandekarKN12}
for the \RestrictedGroupSNDP on directed graphs.
However, the same construction applies for the
{\sf Vertex-Weighted} \RestrictedGroupSNDP. 
We are aware that this fact might not be clear for the readers.
Thus, we provide the sketch of the proof in \Cref{app:hardness}.

We remark that the $k$-CSP hardness implies that even for $k_i\in\{0,2\}$, 
{\sf Vertex-Weighted} \RestrictedGroupSNDP cannot be approximated to within a factor of
$2^{\log^{1-\epsilon}n}$, for any $\epsilon>0$,
unless $\NP\subseteq\DTIME(n^{\polylog{n}})$,
and the approximation hardness is conjectured to be polynomial on $n$, say
$n^{\delta}$ for some $0<\delta<1$, 
under the {\em Sliding Scale Conjecture} \cite{BellareGLR93}.
So far, we do not know of any non-trivial approximation for
this problem for $k\geq 2$.

The edge-cost variant has been studied in \cite{KhandekarKN12,GuptaKR10},
and a polylogarithmic approximation is known for the case $k=2$ \cite{KhandekarKN12}.
For $k>3$, there is no known non-trivial approximation algorithm.
The relaxed variant where the $k$ disjoint paths from the root may
end at different vertices in each group $S_i$ has also been studied in \cite{KhandekarKN12,GuptaKR10}.
Chalermsook, Grandoni and Laekhanukit proposed a bicriteria approximation
algorithm for the {\sf Relaxed} \RestrictedGroupSNDP \cite{ChalermsookGL15}; however, their technique 
is not applicable for the restricted version.
Note that the hardness of the edge-cost variant of {\sf Relaxed} \RestrictedGroupSNDP
is $k^{1/6-\epsilon}$, for any $\epsilon>0$ \cite{ChalermsookGL15}.
It is not hard to construct
the same hardness result for \RestrictedGroupSNDP.
We believe that \RestrictedGroupSNDP is strictly harder
than the relaxed variant.

\subsection{Our Results \& Techniques}

Our main result is the following approximation result for \RestrictedGroupSNDP.

\begin{theorem} 
  \label{thm:main:rgsndp}
  There is an $O(\log n \log h)$ approximation for \RestrictedGroupSNDP  that runs in time $n^{f(w,k)}$ for some function $f$.  
\end{theorem} 

The proof of this theorem relies on the technique introduced in~\cite{ChalermsookDLV17}.
We give an overview of this technique and highlight how this paper departs from it. 

In short, this technique ``bridges'' the ideas of dynamic program (DP) and randomized LP rounding in two steps\footnote{One may view our result as a ``tree-embedding'' type result. Please see \cite{ChalermsookDLV17} for more discussion along this line. Here we choose to present our result in the viewpoint of DP \& LP.}. 
Let us say that we would like to approximate optimization problem $\Pi$. 
In the first step, a ``nice'' DP table that captures the computation of the optimal solution is created, and there is a 1-to-1 correspondence between the DP solution and the solution to the problem.
However, since the problem is NP-hard (in our case, even hard to approximate to within some poly-logarithmic factor), we could not follow the standard bottom-up computation of DP solutions. 
The idea of the second step is to instead write an LP relaxation that captures the computation of the optimal DP solution, and then use a randomized dependent rounding to get an approximate solution instead; the randomized rounding scheme is simply the well-known GKR rounding~\cite{GargKR00}.  
Roughly speaking, the size of the DP table is $n \cdot w^{O(w)}$, while the LP relaxation has $n^{O(w \log w)}$ variables and constraints, so we could get an $O(\log n \log k)$ approximation in time $n^{O(w \log w)}$. 

The main technical hurdle that prevents us from using this technique to \RestrictedGroupSNDP directly 
is that there was no systematic way to generate a ``good'' DP table for arbitrary connectivity demand $k$.
(The previous result was already complicated even for $k=1$.)  
This is where we need to depart from the previous work. 
We devise a new concept that allows us to systematically create such a DP table for any connectivity demand $k$. 
Our DP table has size $n \cdot f(k,w)$ for some function $k$ and $w$, 
and it admits the same randomized rounding scheme in time $n^{g(k,w)}$, therefore yielding the main result. 

As by-products, we obtain new algorithms for some well-studied variants of \sndp, 
whose running time depends on the treewidth of the input graph (in particular, $n \cdot f(k,w)$). 

\shortparagraph{Subset connectivity problems:} 
{\sf Subset $k$-Connectivity} is a well-studied \sndp problem
(\SubsetkEC and \SubsetkVC for edge and vertex connectivity, respectively). 
In this setting, all pairs of terminals have the same demands, i.e., $k(u,v) = k$ for all $u,v \in T$. 
This is a natural generalization of Steiner tree that has received  attention~\cite{CheriyanV07,Nutov12,Laekhanukit15}. 

\begin{theorem} 
There are exact algorithms for \SubsetkEC and \SubsetkVC that run in time $f_1(k,w) n$ for some function $f_1$. This result holds for vertex- or edge-costs.  
\end{theorem}

\shortparagraph{Rooted SNDP:} Another setting that has been studied in the context of vertex connectivity requirements is the \RootedSNDP{}~\cite{ChuzhoyK08,Nutov12}. In this problem, there is a designated terminal vertex $r \in T$, and all positive connectivity requirements are enforced only between $r$ and other terminals, i.e. $k(u,v) >0$ only if $u=  r$ or $v=r$.  
For the edge connectivity setting, \RootedSNDP{} captures \SubsetkEC\footnote{This is due to the transitivity of edge-connectivity. Specifically, any vertices $u,w$ that have $k$ edge-disjoint paths connecting to the root $r$ also have $k$ edge-disjoint paths between themselves.}

\begin{theorem} 
There are exact algorithms for {\sf Rooted} \ecsndp and {\sf Rooted} \vcsndp that run in time $f_2(k,w) n$ for some function $f_2$. This result holds for costs on vertices or edges.  
\end{theorem}

\shortparagraph{Further technical overview:} Let us illustrate how our approach is used to generate the DP table, amenable for randomized rounding. 
The following discussion assumes a certain familiarity with the notion of treewidth and DP algorithms in low-treewidth graphs.  

Given graph $G= (V,E)$, let $\tset$ be a tree decomposition of $G$ having width $w$, i.e. each bag $t \in V(\tset)$ corresponds to a subset $X_t \subseteq V(G): |X_t| \leq w$. Let $\tset_t$ denote the subtree of $\tset$ rooted at $t$.  
For each bag $t \in V(\tset)$, let $G_t$ denote the subgraph induced on all bags belonging to the subtree of $\tset$ rooted at $t$, i.e. $G_t = G[\bigcup_{t \in \tset_t} X_t]$.  
At a high level, DPs for minimization problems in low-treewidth graphs proceed as follows. 
For each ``bag'' $t$, there is a profile $\pi_t \in \Pi $ for $t$, and we define a DP cell $c[t, \pi_t]$ for each possible such profile, which stores the minimum-cost of a solution (a subgraph of $G_t$) that is consistent with the profile $\pi_t$.  
Then, a recursive rule is applied: Let $t', t''$ be the left and right children of $t$ in $\tset$ respectively. 
The DP makes a choice to ``buy'' a subset of edges  $Y \subseteq E(G[X_t])$ (that appear in bag $X_t$) and derives the cost by minimizing over all profiles $\pi_{t'}, \pi_{t''}$ that are ``consistent'' with $\pi_t$:  
\[c[t, \pi_t] = \min_{\pi_{t'}, \pi_{t''}, Y: (\pi_{t'}, \pi_{t''}, Y) \bowtie \pi_t} \left( \mathrm{cost}(Y) + c[t',\pi_{t'}] + c[t'', \pi_{t''}] \right) \]   
where the sign $(\pi_{t'}, \pi_{t''}, Y) \bowtie \pi_t$ represents the notion of consistency between the profiles.
Different optimization problems have different profiles and consistency rules.  
Often, consistency rules that are designed for connectivity-1 problems (such as Steiner tree) are not easily generalizable to higher connectivity problems (such as \sndp).

\newcommand{\consrule}[0]{{\Bumpeq}}
In this paper, we devise a new consistency rule (abbreviated by $\consrule{}$) for checking ``reachability'' (or connectivity 1) in a graph, which allows for easy generalization to handle high connectivity problems.  

Roughly speaking, our consistency rule $\consrule{}$ solves the Steiner tree problem (connectivity-1 problem).  
To solve a connectivity-$k$ problem, we have a DP cell $c[t, \vec{\pi}]$ for each $\vec{\pi} \in \Pi^k$.  
Then the consistency check is a ``direct product'' test for all coordinates, i.e., 
\[(\vec{\pi}_{t'}, \vec{\pi}_{t''},\vec{Y}) \consrule{}^k  \vec{\pi}_t \Longleftrightarrow 
(\forall j \in [k]) (\pi_{t',j}, \pi_{t'',j}, Y_j) \consrule{} \pi_{t,j}  
\] 
In this way, our new concept makes it a relatively simple task to generalize a DP for connectivity-1 problems to a DP for connectivity-k problems (and facilitate the proof of correctness).  
There is a slight change in the way DPs are designed for each problem, but they follow the same principle.

\shortparagraph{Organization:} We develop our techniques over several sections, and along the way, show non-trivial applications for various \sndp problems.
\Cref{sec:prelim} provides some notation and important definitions. 
\Cref{sec:new-concept} presents the new viewpoint for designing DP and presents a simple showcase by deriving (known) results. 
\Cref{sec:ec-sndp} presents algorithms for \ecsndp.
Lastly, \Cref{sec:rgsndp} presents an approximation algorithm for group connectivity problems.

\section{Preliminaries}
\label{sec:prelim}
\shortparagraph{Tree decomposition:} Let $G$ be any graph.
A {\em tree decomposition} of $G$ is a tree $\tset$ with a collection of {\em bags} $\{X_t\}_{t \in V(\tset)}\subseteq 2^{V(G)}$ (i.e., each node of $\tset$ is associated with a subset of nodes of $V(G)$) that satisfies the following properties: 

\begin{itemize} 
\item $V(G) = \bigcup_{t\in V(\tset)} X_t$ 
\item For any edge $uv \in E(G)$, there is a bag $X_t$ such that $u,v \in X_t$.

\item For each vertex $v \in V(G)$, the collection of nodes $t$ whose bags $X_t$ contain $v$ induces a connected subgraph of $\tset$. That is, $\tset[\{t\in V(\tset): v\in X_t\}]$ is a subtree of $\tset$.
\end{itemize}  

The treewidth of $G$, denoted $tw(G)$, is the minimum integer $k$ for which there exists a tree decomposition $(\tset, \{X_t\}_{t \in V(\tset)})$ such that $\max |X_t| \leq k+1$ ($\max |X_t|-1$ is the \emph{width} of $\tset$).

Fix a tree decomposition $(\tset, \{X_t\})$ with the stated properties. 
For each node $t \in V(\tset)$, denote by $\tset_t$ the subtree of $\tset$ rooted at $t$. 
We also define $G_t$ as the subgraph induced by $\tset_t$; that is, $G_t = G\left[\bigcup_{t' \in \tset_t} X_{t'}\right]$. 

For each $v \in V$, let $t_v$ denote the topmost bag for which $v \in X_{t_v}$. 
For each $t \in V(\tset)$, we say that an edge $uv \in E(G)$ appears in the bag $t$ if $u,v \in X_t$, and \emph{only} if $t$ is the topmost bag in which this happens.
We denote the edges inside the bag $X_t$ by $E_t$. 
For a subset of bags, $\sset \subseteq V(\tset)$, we define $X_{\sset} := \bigcup_{t \in \mathcal S} X_t$. 

We will use the following result, which shows that a tree decomposition of $G$ of width $O(tw(G))$ is computable in time $O(2^{O(tw(G))} n)$. 

\begin{theorem}[\cite{BodlaenderDDFLP16}] 
\label{thm: good decomp} 
There is an algorithm that, given a graph $G$, runs in time $O(2^{O(tw(G))} n)$ and finds the tree decomposition $(\tset, \{X_t\}_{t \in V(\tset)})$ such that $|X_t| \leq 5 tw(G)$ for all $t$.  
\end{theorem}

In order to simplify notation, we assume that $\tset$ is a binary tree. Furthermore, we require the height of $\tset$ to be $O(\log n)$ in \Cref{sec:rgsndp}. The following lemma, based on the results of Bodlaender~\cite{Bodlaender88}, summarizes the properties we assume.

\begin{lemma}[in \cite{ChalermsookDLV17}, based on \cite{Bodlaender88}]
\label[lemma]{lem:treewidth:props}
There is a tree decomposition $(\tset, \{X_t\}_{t \in V(\tset)})$ with the following properties:
\begin{inparaenum}[(i)]
\item the height of $\tset$ is at most $O(\log n)$; 
\item each bag $X_t$ satisfies $|X_t| \leq O(w)$; 
\item every leaf bag has no edges ($E_t = \emptyset$ for leaf $t \in \tset$); 
\item every non-leaf has exactly $2$ children 
\end{inparaenum}
\end{lemma}

\shortparagraph{Connection sets and operators:} Let $S \subseteq V$. 
A connection set $\Lambda$ over $S$ is a subset of $S \times S$ which will be used to list all pairs that are connected via a path, i.e., $(u,v) \in \Lambda$ iff there is a path connecting $u$ to $v$.

\begin{figure}
  \vspace{-3.8em}
  \begin{center}
    \includegraphics[trim={2cm 6cm 0 0},clip,width=0.48\textwidth]{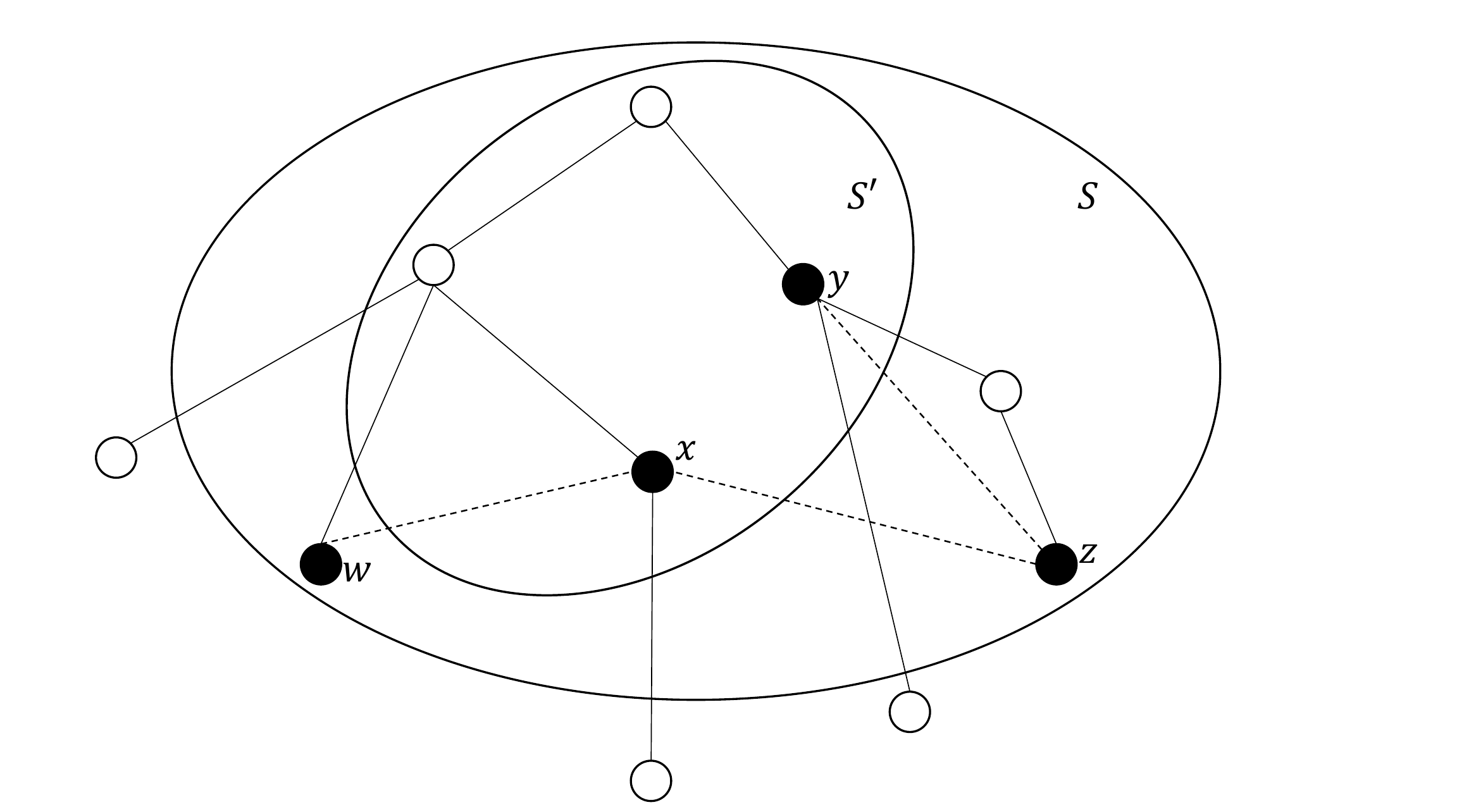} 
  \end{center}

  \caption{Connection sets, transitive closures and projections. \newline{}
      \mbox{$\Lambda = \{(w,x), (x,z), (z,y)\}$} (dotted connections).
      \mbox{$\tc(\Lambda) = \{(w,x), (x,y), (y,z), (w, y), (x, z), (w, z)\}$.} 
      \mbox{$\restr{\tc(\Lambda)}_{S'} = \{(x, y)\}$}. 
  }
  \label{figf:connection}
\end{figure}

Let $\Lambda$ be a connection set over $S$. 
The {\em transitive closure} operator, denoted by $\tc(\cdot)$, is defined naturally such that $\tc(\Lambda)$ contains all pairs $(w,w')$ for which there is a sequence $(w=w_0,w_1,\ldots, w_q = w')$ and $(w_i, w_{i+1}) \in \Lambda$ for all $i <q$.  
Let $S' \subseteq S$. The {\em projection operator} ``$\mid$'' is defined such that $\restr{\Lambda}_{S'}= \Lambda \cap (S' \times S')$. See \Cref{figf:connection} for an illustration.  
   
Given two connection sets $\Lambda_1$ of $S_1$ and $\Lambda_2$ of $S_2$, the union $\Lambda_1 \cup \Lambda_2$ is a connection set over $S_1 \cup S_2$.

\vspace{0.3em}

\section{New Key Concept: Global \texorpdfstring{$\Leftrightarrow$}{^^e2^^87^^94} Local Checking for DP}
\label{sec:new-concept}

This section introduces the key concept devised for handling all our problems systematically.

\shortparagraph{High-level intuition:} Our DP will try to maintain a pair of local and global information about connectivity in the graph.
Roughly speaking, a local connection set $\Gamma_t$ for $t$ (more precisely, for $X_t$) gives information about connectivity of the solution inside the subgraph $G_t$ (the subgraph induced in subtree $\tset_t$), while the other connection $\Delta_t$ for $t$ gives information about connectivity of the global solution (i.e., the solution for the whole graph $G$).  
 
For instance, if we have a tentative solution $Y \subseteq E(G)$, we would like to have the information about the reachability of $Y$ inside each bag, i.e., $\Delta_t = \restr{\tc(Y)}_t$, so we could check the reachability between $u$ and $v$ simply by looking at whether $(u,v) \in \Delta_t$. 
However, a DP that is executing at bag $t$ may not have this global information, and this often leads to complicated rules to handle this situation. 

We observe that global information can be passed along to all cells in the DP with simple local rules so that checking whether the connectivity requirements are satisfied can be done locally inside each DP cell.  
In the next section, we elaborate on this more formally.

\shortparagraph{Equivalence:} One could imagine having a DP cell $c[t, \Gamma_t,\Delta_t]$ for all possible connection sets $\Gamma_t,\Delta_t$, which makes a decision on $Y_t$, the set of edges bought by the solution when executing the DP.
The roles of $\Gamma_t$ and $\Delta_t$ are to give information about local and global reachability, respectively.  
We are seeking a solution $Y$ that is ``consistent'' with these profiles, where $Y$ can be partitioned based on the tree $\tset$ as $Y_t = Y \cap E_t$.   

Given $Y$, we say that the pairs $\{(\Gamma_t, \Delta_t)\}_{t \in V(\tset)}$ satisfy the \emph{local} (resp., \emph{global})
\emph{connectivity definition} if, for every bag $t \in V(\tset)$ (having left and right children as $t'$ and $t''$, respectively),

\noindent
\begin{minipage}[t]{.5\textwidth}
   \vspace{-2em}
   \begin{align*}
      \text{\bf Local}& \\
      \Gamma_t &:= 
      \begin{cases}
         \emptyset                                          & \text{if $t$ is a leaf bag} \\
         \restr{\tc\left(\Gamma_{t'}\cup\Gamma_{t''}\cup Y_t\right)}_t  & \text{otherwise} 
      \end{cases} \\
      \Delta_t &:= 
      \begin{cases}
         \Gamma_t                                        & \text{if } t = \rootn(\tset) \\
         \restr{\tc(\Delta_{p(t)}\cup\Gamma_t)}_t         ~~~~~~& \text{otherwise} 
      \end{cases} 
   \end{align*}
\end{minipage}%
\begin{minipage}[t]{.5\textwidth}
   \vspace{-2em}
   \begin{align*}
      \text{\bf Global}& \\
      \Gamma_t &:= \restr{\tc\left(Y \cap E_{\tset_t}\right)}_t \\
      \Delta_t &:= \restr{\tc\left(Y\right)}_t 
   \end{align*}
\end{minipage}

where the projection operator on $X_t$ is simplified as $\mid_{t}$. 

The main idea is that the local connectivity definition gives us ``local'' rules that enforce consistency of consecutive bags, and this would be suitable for being embedded into a DP.  
The global connectivity rules, however, are not easily encoded into DP, but it is easy to argue intuitively and formally about their properties. 
The following lemma (proof in \Cref{appendix:proofs:unify}) shows that the local and global connectivity definitions are, in fact, equivalent.
\begin{lemma}
\label[lemma]{lem:edge:unify}
Let $Y_t \subseteq E_t$ be a subset of edges and $(\Gamma_t, \Delta_t)$ a pair of connectivity sets for every $t \in V(\tset)$. %
Then, the pairs $(\Gamma_t, \Delta_t)$ satisfy the local connectivity definition iff they satisfy the global connectivity definition.
\end{lemma}

The notions of local and global connectivity, as well as the equivalence
between them can be generalized both for the edge-connectivity version with
vertex-costs as well as vertex-connectivity with vertex-costs. %
We defer the details of this generalization to \Cref{sub:vcnw:extension}.

\shortparagraph{A warmup application: steiner trees.} We now show an approach that allows us to solve the Steiner Tree problem
exactly in $n w^{O(w)}$ time, given a tree decomposition of width $w$.
We remark that the best known algorithm due to Cygan et al.~\cite{CyganNPPRW11} runs in time $2^{O(w)} n$.  
In this problem, we are given graph $G=(V,E)$ with edge-costs and terminals $T = \{v_1,\ldots, v_h\}$, and the goal is to find a min-cost subset $E^* \subseteq E$ that connects all the terminals. 
For simplicity, we denote $v_1$ by ``root'' $r$, and the goal is to connect the root to all other terminals in $T \setminus r$.

Our DP table has a cell $c[t, \Gamma, \Delta]$ for every bag $t \in V(\tset)$
and every pair of connection sets $(\Gamma, \Delta)$ for $t$. %
We initialize the DP table by setting $c[t, \Gamma, \Delta]$ for all the leaf
bags, and setting certain cells as invalid (by setting $c[t, \Gamma, \Delta] = \infty)$:
\begin{itemize}
   \item For every leaf bag $t$ and every pair $(\Gamma, \Delta)$, 
   we set $c[t, \Gamma, \Delta] = 0$ if $\Gamma = \emptyset$ and $c[t,
   \Gamma, \Delta] = \infty$ otherwise.

   \item We mark the cells $(\rootn(\tset), \Gamma, \Delta)$ as invalid if
   $\Gamma \neq \Delta$.

   \item Let $v_i \in T$ be one of the terminals, and $t \in V(\tset)$ a bag. We
   mark a cell $(t, \Gamma, \Delta)$ as invalid if $v_i \in X_t$ but $(r, v_i)
   \not\in \Delta$.
\end{itemize}

For all other cells, we compute the bags from their children. %
Let $t$ be a bag with left-child $t'$ and right-child $t''$.
Let $(\Gamma,\Delta), (\Gamma', \Delta'), (\Gamma'', \Delta'')$ be pairs of connection sets for
$t, t', t''$, respectively, and $Y_t \subseteq E_t$. %
We say that $(\Gamma, \Delta)$ is consistent with $((\Gamma', \Delta'), (\Gamma'', \Delta''))$  via $Y_t$ 
(abbreviated by the notation $(\Gamma, \Delta) \cons{Y_t} ((\Gamma', \Delta'), (\Gamma'', \Delta''))$) if
\begin{align*}
\Gamma &= \restr{\tc(\Gamma' \cup \Gamma'' \cup Y)}_t &
\Delta' &= \restr{\tc(\Delta \cup \Gamma')}_{t'} &
\Delta'' &= \restr{\tc(\Delta \cup \Gamma'')}_{t''}
\end{align*}
Now, for any choice of valid DP cells $(t, \Gamma_t, \Delta_t)$ and edge subsets
$Y_t \subseteq E_t$ for every $t \in V(\tset)$ %
(notice that a DP solution uses precisely one cell per bag $t$), %
we apply \Cref{lem:edge:unify} to conclude that since the local
connectivity definition is satisfied for $(\Gamma_t, \Delta_t)$ pairs, so does
the global connectivity definition. %
Since, for every valid DP cell $(t, \Gamma, \Delta)$ such that $t$ contains
a terminal $v_i$, $(r, v_i) \in \Delta$, we conclude that every terminal is
connected to the root in the solution $Y = \bigcup_{t \in V(\tset)} Y_t$.

Conversely, given a solution $F \subseteq E(G)$, we can define $F_t := F \cap
E_t$, and a pair $(\Gamma_t, \Delta_t)$ for every $t \in V(\tset)$, 
using the global connectivity definition. %
\Cref{lem:edge:unify} implies that the pairs $(\Gamma_t, \Delta_t)$
satisfy the local connectivity definition, and therefore define valid DP cells
(notice that for every terminal $v_i$ , $F$ connects $v_i$ to the root, so
$(r, v_i) \in \Delta_t$ for every $t\in V(\tset)$ such that $v_i \in X_t$).
We thus establish that for every valid DP solution there is a corresponding
feasible solution \mbox{$F \subseteq E(G)$}, and vice-versa.

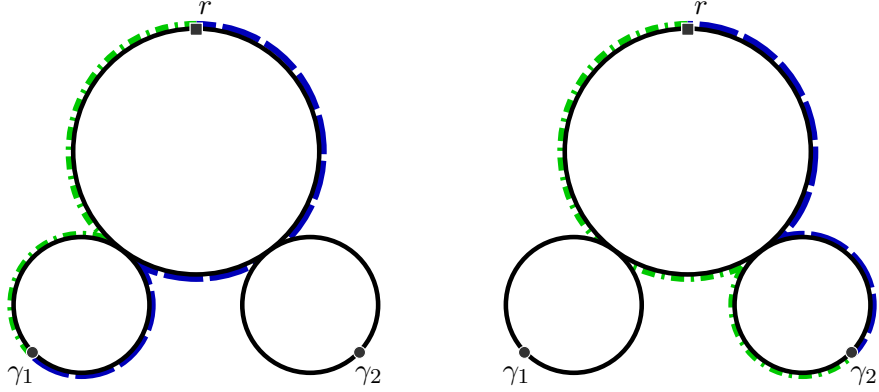
\begin{figure}[t]
   \centering
   \begin{overpic}[width=0.7\textwidth]{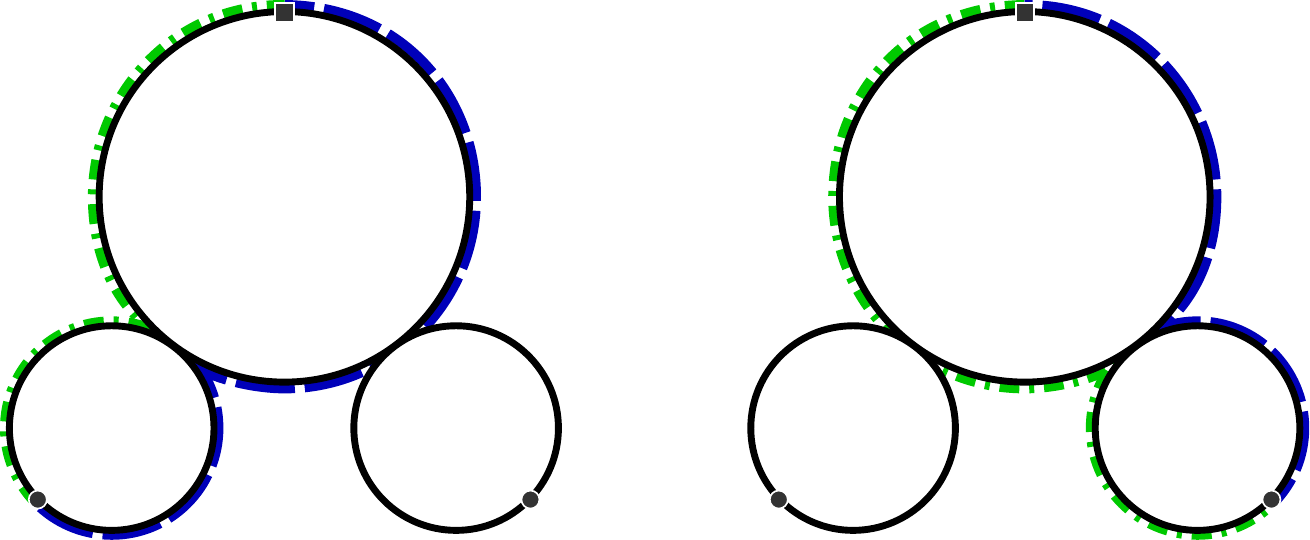}
      \put(79, 42) {$r$}
      \put(22, 42) {$r$}
      \put(0, 0) {$\gamma_1$}
      \put(40, 0) {$\gamma_2$}
      \put(57, 0) {$\gamma_1$}
      \put(97, 0) {$\gamma_2$}
   \end{overpic}
   \caption{Example of different partitions of edges into paths from demands
   $\gamma_1$, $\gamma_2$. On the left (resp., right), two paths from $r$ to
   $\gamma_1$ (resp., $\gamma_2$)}
   \label{fig:multiple-partitions}
\end{figure}

\section{Extension to High Connectivity}
\label{sec:ec-sndp} 

This section shows how to apply our framework to 
problems with high connectivity requirements.
We focus on edge-connectivity and leave 
the case of vertex-connectivity to \Cref{sub:vcnw:extension}.

When solving a problem in a high connectivity setting, there may be a
requirement of $k$ disjoint paths. %
In particular, for each demand pair $(u,v)$, there must be $k(u,v)$ disjoint
paths in the solution. So we could start naturally with a profile of the form:
\[(\Gamma_{t,1},\ldots, \Gamma_{t,k})(\Delta_{t,1},\ldots, \Delta_{t,k}) \] 
for each bag $t \in V(\tset)$, and enforce the local consistency conditions for each coordinate. %

However, this idea does not work as a different demand pair, say $(a,b)$, might use a path that belongs to different subgraphs $H_j$ as defined above. In other words, the disjoint paths for the demand pair $(a, b)$ might require a different partitioning of the solution set $H$. 
\Cref{fig:multiple-partitions} illustrates a case in which different demand pairs use different partitions. %
Therefore, we need to enumerate all possible ways for the demands to ``locally'' partition the graph and use them to support all $k$ disjoint paths for each one of them. This requires a more careful local consistency check between the DP cells. 

We consider the \RootedECSNDP problem with edge-costs as an example to explain how to apply our framework in high-connectivity. %
For convenience, we add $r$ to every bag. %
To avoid confusion, the root of the tree $\tset$ will be referred explicitly as $\rootn(\tset)$. 

The organization of this section is as follows. 
In \Cref{sec:EC-profile}, we explain the setup of the cells of the DP table and a high-level intuition about how the DP works.  
In \Cref{sub:ecew:dp} we describe the algorithm, in particular, how to compute the values of the DP table.
We leave the discussion of its correctness and running time to \Cref{sub:ecew:running_time}.

\subsection{Profiles}
\label{sec:EC-profile}

As in any standard dynamic programming approach based on tree decomposition, we have a profile for each bag $t$, which tries to solve the subproblem restricted to $G_t$ in some way. 

Let $t \in V(\tset)$ be a bag. 
A {\em connection profile} for $t$ is a k-tuple $\vec{\Gamma}= (\Gamma_1,\Gamma_2,\ldots, \Gamma_k)$ such that $\Gamma_i \subseteq X_t \times X_t$.
Let $\xset_t$ be the set of all connection profiles for $t$.  
A profile $\Psi$ of node $t$ is a collection of pairs of connection profiles $(\vec{\Gamma}, \vec{\Delta})$, i.e., $\Psi \subseteq \xset_t \times \xset_t$.  
A partial solution $F \subseteq E(G_t)$ is said to be {\em consistent} with profile $\Psi$ for $t$ if, for all $(\vec{\Gamma}, \vec{\Delta}) \in \Psi$,   
\begin{itemize} 
\item For each $(u,v) \in \Gamma_i$ and $(a,b) \in \Gamma_j$ for $i\neq j$, there are paths $P_{uv}, P_{ab} \subseteq F$ connecting the respective vertices in $G_t$ such that $P_{uv}$ and $P_{ab}$ are edge-disjoint. 
\item There is a global solution $F' \supseteq F$ such that, for each $(u,v) \in \Delta_i$ and $(a,b) \in \Delta_j$ for $i\neq j$, there are paths $Q_{uv}, Q_{ab} \subseteq F'$ connecting the respective vertices in $G$ such that $P_{uv}$ and $P_{ab}$ are edge-disjoint.  
\end{itemize}  
In other words, a solution consistent with a profile must ``implement'' all connectivity requirements by $\vec{\Gamma}$ and must be extensible to satisfy $\vec{\Delta}$.  

Passing down both local and global requirements in the DP table leads to a clean and simple DP algorithm.
Our DP table has a cell $c[t,\Psi]$ for each bag $t \in V(\tset)$ and each profile $\Psi$ for $t$. 
This cell tentatively stores the optimal cost of a solution consistent with profile $\Psi$.

\subsection{The DP}  
\label{sub:ecew:dp}

\shortparagraph{Valid cells:} Some table entries do not correspond to valid solutions of the problem, so we mark them as invalid and remove them from consideration (another way to think about this is that we initialize $c[t, \Psi] = \infty$ for all invalid cells), i.e., the following cells are invalid: 
\begin{itemize} 
\item Any leaf that has non-empty connectivity requirements is invalid. That is, $c[t, \Psi] = 0$ if $\Psi \subseteq \{(\emptyset, \ldots, \emptyset)\} \times \xset_t$; otherwise, $c[t, \Psi] = \infty$.  

\item Any cell that cannot be extended into a feasible solution is invalid. 
If there is no pair $(\vec{\Gamma}, \vec{\Delta}) \in \Psi$ such that $(r, \gamma_i) \in \Delta_j$ for all $j \in [k_i]$, 
then there are fewer than $k_i$ edge-disjoint paths between $r$ and $\gamma_i$, and therefore the cell is invalid, so $c[t_{\gamma_i},\Psi] = \infty$. 

\item The bag $\rootn(\tset)$ together with profile $\Psi$ is an invalid cell if there is a pair $(\vec{\Gamma}, \vec{\Delta})\in \Psi$ such that $\Gamma_j \neq \Delta_j$ for some $j$. In this case, we set $c[\rootn(\tset), \Psi] = \infty$.
\end{itemize}  

\begin{lemma}
\label[lemma]{lem:ecew:dpsize}
For every $t \in \tset$, there are at most $\exp(w^{O(wk)})$ many valid cells $(t, \Psi)$.
\end{lemma}

\shortparagraph{DP computation:} For all other cells, we compute their values from the values of their children. 
Let $t$ be a bag with left-child $t'$ and right-child $t''$. 
Let $\Psi, \Psi', \Psi''$ be their profiles respectively, and $Y \subseteq E_t$. 
We say that $\Psi$ is consistent with $(\Psi', \Psi'')$ via $Y$ (abbreviated by $\Psi \cons{Y} (\Psi', \Psi'')$) if the following conditions are satisfied. For each pair $(\vec{\Gamma}, \vec{\Delta}) \in \Psi$, there are $(\vec{\Gamma}' , \vec{\Delta}') \in \Psi'$ and  $(\vec{\Gamma}'' , \vec{\Delta}'') \in \Psi''$, together with a partition of $Y$ into $Y_1 \cup Y_2 \cup \ldots \cup Y_k$ such that, for every $j \in [k]$,
\begin{align*}
\Gamma_j &= \restr{\tc(\Gamma'_j \cup \Gamma''_j \cup Y_j)}_t &
\Delta'_j &= \restr{\tc(\Delta_j \cup \Gamma'_j)}_{t'} &
\Delta''_j &= \restr{\tc(\Delta_j \cup \Gamma''_j)}_{t''}
\end{align*}
Similarly, for any $(\vec{\Gamma'}, \vec{\Delta'}) \in \Psi'$ (resp., $(\vec{\Gamma''}, \vec{\Delta''}) \in \Psi''$)  there are $(\vec{\Gamma''}, \vec{\Delta''}) \in \Psi''$ (resp., $(\vec{\Gamma'}, \vec{\Delta'}) \in \Psi'$) plus $(\vec{\Gamma}, \vec{\Delta}) \in \Psi$ and some partition of $Y$,  satisfying similar conditions as above.

Then the value of $c[t, \Psi]$ can be defined recursively among valid cells: 
\[c[t, \Psi] = \min_{\Psi \cons{Y} (\Psi', \Psi'')} \left( c[t', \Psi'] + c[t'', \Psi''] + c(Y) \right) \]
The final solution can be computed as 
$\min \left\{ c[\rootn(\tset),\Psi] \mid (\forall (\vec{\Gamma}, \vec{\Delta}) \in \Psi)(\forall j \in [k]) \Gamma_j = \Delta_j \right\}$.
The correctness of this DP is deferred to \Cref{sub:ecew:correctness}.

\section{Algorithms for \RestrictedGroupSNDP} %
\label{sec:rgsndp}
We now have sufficient technical tools to prove \Cref{thm:main:rgsndp}. In the first step, we turn the DP table into a tree instance of a variant of \gst and, in the second step, apply randomized rounding to obtain a polylogarithmic approximation to the problem.

\newcommand{\fiber}[0]{\ensuremath{\operatorname{copies}}}

\shortparagraph{Tree Instance: } %
We start by showing how to
 transform the DP table into a tree $\widetilde \tset$, where we can solve a
variant of the group Steiner tree problem. %
The following theorem formalizes this transformation, and we dedicate the rest
of this section to proving the theorem.
For convenience, we add a dummy bag $t_S$ with $X_{t_S} = \emptyset$ as the parent
of the root bag. %

\begin{theorem}
\label{thm:rgsndp:tree_instance}

Given a graph $G$ rooted at $r$ with treewidth $w$ and groups $S_i \subseteq V$, there is a
tree $\widetilde \tset$ with groups $\widetilde S_i$ and a set of accepted solutions $\widetilde X$ such that:
\begin{inparaenum}[(i)]
     \item the size of $\widetilde \tset$ is $n^{w^{O(wk)}}$;
     \item for every $F \subseteq E(G)$, there is $X \in \widetilde X$ (and
     vice-versa) such that $c(F) = c(X)$ and, for every $i \in [h]$, $F$
     $k_i$-connects $r$ to $v \in S_i$ iff $X$ connects $\rootn(\widetilde T)$
     to $\ttil \in \widetilde S_i$.
\end{inparaenum}
\end{theorem}
For each cell of the DP table introduced in \Cref{sub:ecew:dp}, we
create a node in $\widetilde \tset$. Namely, we create a vertex $\ttil[t, \Psi]$
for every bag $t \in V(\tset)$ and $\Psi \subseteq \xset_t \times
\xset_t$. %
The root of the tree is $\ttil[t_S, \{(\emptyset^{k},\emptyset^{k})\}]$ (this is
the only connection profile for $t_S$). %
For a bag $t \in \tset$ with children $t'$, $t''$, %
we add connecting nodes $\ttil_c[t,\Psi, \Psi', \Psi'', Y]$ connected to
the nodes $\ttil[t, \Psi]$, $\ttil[t', \Psi']$, $\ttil[t'', \Psi'']$, for every $Y
\subseteq E_t$, if $\Psi \cons{Y} (\Psi', \Psi'')$. %
If there is only one child, the connecting node has degree $2$, and we
consider that $\Psi'' = \{(\emptyset^{k},\emptyset^{k})\}$ for the purpose of
describing the algorithm. %
An edge from $\ttil[t, \Psi]$ to $\ttil_c[t,\Psi, \Psi', \Psi'', Y]$ is labeled with
the set of edges $Y$ and is assigned cost $c(Y)$. All other edges in the
instance have cost $0$. 

Notice that, at this point, $\widetilde \tset$ is not a tree, but we can turn
it into one by making copies of the nodes as required. Specifically, we
process the tree in a bottom-up fashion: for each node $\ttil$, we make the 
same number of copies of $\ttil$ and its descendants as there are incoming
edges of $\ttil$ such that each edge is incident to a different copy. 
In this manner, all the copies of $\ttil$ are now the roots of
subtrees, which are disjoint. %

For convenience, we denote by $\ttil[t, \Psi]$ (resp., $\ttil_c[t,\Psi, \Psi',
\Psi'', Y]$) any copy of the original node; when we need to
distinguish copies, we denote by $\fiber(\ttil)$ the set of all
copies of a node $\ttil$. %

The final step in our construction is to prune the tree by removing nodes
that cannot be reached or that represent choices that cannot be part of a
feasible solution. To do that, it is sufficient to apply the following rules
to exhaustion:
\begin{inparaenum}[(i)]
     \item remove $\ttil \in \widetilde \tset$ if it is not connected to the root;
     \item remove a connecting node if one of its children was removed;
     \item remove $\ttil[t, \Psi]$ if it is a leaf node but  $\Psi \not\subseteq \{(\emptyset, \ldots, \emptyset)\} \times \xset_t$  (i.e.~$\vec\Gamma \neq (\emptyset, \ldots, \emptyset)$ for some $(\vec\Gamma, \vec\Delta) \in\Psi$).
\end{inparaenum}

We can now restate the goal of the problem in terms of $\widetilde \tset$: we
want to find nodes $\ttil[t, \Psi_t]$ and edge sets $Y_t \subseteq E_t$ for
every bag $t \in V(\tset)$, such that $(\Psi_{t'}, \Psi_{t''}) \cons{Y}
 \Psi_t$ for all non-leaf bags $t$ with children $t'$, $t''$. %
Further, for every group $S_i$, there must be a vertex
$\gamma_i \in S_i$ and a partition of $Y := \bigcup_{t \in
V(\tset)} Y_t$ into $k_i$ sets, such that each contains a path from $r$ to $\gamma_i$. %

The set of nodes $\{\ttil[t, \Psi_t]\}_{t \in V(\tset)}$ with the
respective connecting nodes $\ttil[t, \Psi_t, \Psi_{t'}, \Psi_{t''}, Y_t]$
(for all non-leaf bags $t \in V(\tset)$) induces a tree $\widetilde T$ in
$\widetilde \tset$. %
We say that such a tree $\widetilde T$ is \emph{valid} if every node $\ttil[t, \Psi]
\in V(\widetilde T)$ has exactly one child in the graph (or none if $t$ is a
leaf), and every connecting node $\ttil_c[t,\Psi, \Psi', \Psi'', Y_t]$ in
$\widetilde T$ has full-degree, i.e., all its neighbors are in the solution as
well. %

For every group $S_i$, we define $\widetilde S_i$ as follows: %
for every $v \in S_i$ and every $\Psi \in \xset_t \times \xset_t$, every element of $\fiber(\ttil[t_v,
\Psi])$ is in $\widetilde S_i$ if there is $(\vec{\Gamma}, \vec{\Delta}) \in \Psi$ such that
$(r,v) \in \Delta_j$ for all $j \in [k_i]$.

The size of the instance follows by considering its height and maximum degree: %
the maximum degree of $\widetilde\tset$ is $\exp(w^{O(wk)})$ by \Cref{lem:ecew:dpsize}, 
and there is a tree decomposition of height
$O(\log n)$ by \Cref{lem:treewidth:props}, which implies that
$\operatorname{height}(\widetilde\tset) = O(\log n)$. %
We conclude that $|\widetilde \tset| = n^{w^{O(wk)}}$.

The correctness of the reduction follows from the correctness of the DP for {\sf Rooted} \ecsndp, and its proof is left to \Cref{sub:ecew:correctness}.

\shortparagraph{Algorithm: } %
We now show how to obtain a valid tree $\widetilde T$, given $\widetilde \tset$. %
Let $T^*$ be the min-cost valid tree in $\widetilde \tset$ that connects all
the groups $\{\widetilde S_i\}_{i \in [h]}$. %
Chalermsook et al.~\cite{ChalermsookDLV17} showed that it is possible to find a valid tree $\widetilde T$ with expected cost $c(T^*)$,
but whose probability of covering a group is just $O(1/\operatorname{height}(\widetilde \tset))$. %

Using this result, we can obtain valid trees $\widetilde T_1$, \ldots, $\widetilde T_\ell$, where $\ell = O(\log n \log h)$. %
We can then obtain solutions $F_j$ that $k$-connect the same groups and have
the same cost as $\widetilde T_j$, for all $j \in [\ell]$, and finally output the solution %
$F := \bigcup_{j \in [\ell]} F_j$. %
Since the expected cost of each $\tilde T_i$ is $c(T^*)$, the expected cost of $F$ is $O(\log n \log h) c(T^*)$.

By sampling $c\log n \log h$ independent valid trees, for large enough $c$, we
ensure that all the groups are covered with high probability (by union bound). %
We conclude that the algorithm outputs a randomized $O(\log n
\log h)$-approximation to the problem, with high probability.

\newcommand{\acksandrefs}[0]{
   \subparagraph*{Acknowledgements.}
   Part of this work was done while Parinya Chalermsook, Bundit Laekhanukit and Daniel Vaz were visiting the Simons Institute for the Theory of Computing. It was partially supported by the DIMACS/Simons Collaboration on Bridging Continuous and Discrete Optimization through NSF grant \#CCF-1740425. 
   Parinya Chalermsook is currently supported by European Research Council (ERC) under the European Union’s Horizon 2020 research and innovation programme (grant agreement No 759557) and by Academy of Finland Research Fellows, under grant number 310415

   \bibliography{main-arxiv}

   \appendix
}
\newcommand{\acksandrefsFull}[0]{
\acksandrefs{}
}
\newcommand{\acksandrefsShort}[0]{
}

\acksandrefsShort{}

\section{Details of the Global \texorpdfstring{$\Leftrightarrow$}{^^e2^^87^^94} Local Checking for DP} %
\label{appendix:proofs:unify}

In this section, we will prove a generalized version of \Cref{lem:edge:unify},
that works for edge-connectivity and vertex-connectivity, both with edge and
vertex costs. %
In vertex-connectivity problems, we are interested in finding internally disjoint
paths. In order to handle this setting, we introduce a modified version of
transitive closure. %

For a set of vertices $Z$ and a set of edges $S$, we denote by $\tc^*_Z(S)$
the set of all pairs $(u,v)$ such that there is a $u$-$v$-path in the graph
$(Z\cup\{u,v\}, S)$, that is, a path whose internal vertices are in $Z$, and
whose edges are in $S$. Formally,
\[
\tc^*_Z(S) = \bigl\{(u,v) \bigm\vert \exists w_1, \ldots, w_\ell \in Z, \forall i \in [\ell-1], (u,w_1), (w_\ell, v), (w_i, w_{i+1}) \in S \big\}
\]
We then keep track, for every bag, of which vertices are allowed to be used in
the solution, that is, we use triples $(Z, \Gamma^*, \Delta^*)$, instead of
the previously used pairs $(\Gamma, \Delta)$. %

Let $W_t \subseteq X_t \setminus X_{p(t)}$ for every $t \in V(\tset)$, $W = \bigcup_{t \in V(\tset)} W_t$, %
$Y_t \subseteq E_t$ for every $t \in V(\tset)$, $Y = \bigcup_{t \in V(\tset)} Y_t$, %
and a triple $(Z_t, \Gamma^*_t, \Delta^*_t)$ for every $t \in V(\tset)$. %
For the purposes of this section, we introduce the following definitions for
local and global connectivity.

We say that the triples $(Z_t,\Gamma^*_t, \Delta^*_t)$ satisfy the \emph{local (resp. global)
connectivity definition} if, for every bag $t \in V(\tset)$,

\noindent
\begin{minipage}[t]{.5\textwidth}
   \vspace{-2em}
   \begin{align*}
      \text{\bf Local}& \\
      Z_t      &:=
      \begin{cases}
         W_t                                    & \text{if } t = \rootn(\tset) \\
         \left(Z_{p(t)} \cup W_t\right) \cap V_t   ~~~~~~& \text{otherwise}
      \end{cases} \\
      \Gamma_t^* &:= 
      \begin{cases}
         \emptyset                                                   & \text{if $t$ is a leaf bag $t$} \\
         \restr{\tc^*_{Z_t}\left(\Gamma^*_{t'}\cup\Gamma^*_{t''}\cup Y_t\right)}_t  & \text{otherwise} 
      \end{cases} \\
      \Delta_t^* &:= 
      \begin{cases}
         \Gamma^*_t                                             & \text{if } t = \rootn(\tset) \\
         \restr{\tc^*_{Z_t}\left(\Delta^*_{p(t)}\cup\Gamma^*_t\right)}_t ~~~~& \text{otherwise} 
      \end{cases} 
   \end{align*}
\end{minipage}%
\begin{minipage}[t]{.5\textwidth}
   \vspace{-2em}
   \begin{align*}
      \text{\bf Global}& \\
      Z_t      &:= W \cap X_t \\
      \Gamma^*_t &:= \restr{\tc^*_W\left(Y \cap E_{\tset_t} \right)}_t \\
      \Delta^*_t &:= \restr{\tc^*_W\left(Y\right)}_t
   \end{align*}
\end{minipage}

\vspace{0.5em}

We then prove the following lemma, proving that the given local and global connectivity definitions are equivalent.
\begin{lemma}
\label[lemma]{lem:mixed:unify}
Let $W_t \subseteq X_t \setminus X_{p(t)}$ be a subset of vertices, $Y_t
\subseteq E_t$ a subset of edges and $(Z, \Gamma^*_t, \Delta^*_t)$ a triple of
profiles for every $t \in V(\tset)$. %

Then, the triples $(Z_t, \Gamma^*_t, \Delta^*_t)$ satisfy the local
connectivity definition iff they satisfy the global connectivity definition.
\end{lemma}

Before proving the lemma, we show how \Cref{lem:edge:unify} follows. %
We will prove that, if we fix $W_t = X_t \setminus X_{p(t)}$ and
$Z_t = X_t$, the definitions of \Cref{lem:mixed:unify,lem:edge:unify} are equivalent. %

For this, it is sufficient to see that $\tc^*_W(S) = \tc^*_{V(G)}(S) = \tc(S)$, and
that the common vertices in $\Gamma^*_{t'}$, $\Gamma^*_{t''}$, and $Y_t$ are
all in $X_t$, thus
\[
\restr{\tc^*_{Z_t}\left(\Gamma^*_{t'}\cup\Gamma^*_{t''}\cup Y_t\right)}_t 
= \restr{\tc\left(\Gamma^*_{t'}\cup\Gamma^*_{t''}\cup Y_t\right)}_t
\]
Similarly, since $\Delta^*_{p(t)}$ and $\Gamma^*_t$ only intersect inside $X_t \times X_t$,
\[
\restr{\tc^*_{Z_t}\left(\Delta^*_{p(t)}\cup\Gamma^*_t\right)}_t
=\restr{\tc\left(\Delta^*_{p(t)}\cup\Gamma^*_t\right)}_t
\]
We conclude that when $Z_t = X_t$, $\Gamma^*_t = \Gamma_t$ and $\Delta^*_t =
\Delta_t$, the proof follows.

The following technical lemma will be useful when proving \Cref{lem:mixed:unify}.
\begin{lemma}[Path Lemma]
\label[lemma]{lemma:monochromatic}
Let $G$ be any graph and $\tset$ be a tree decomposition of $G$. %
Let $t \in V(\tset)$ be a bag and $P$ be a path of length at least $2$ whose
endpoints $x,y$ are the only vertices of $P$ in $t$, that is, $V(P) \cap X_t \subseteq
\{x,y\}$.

Then there is a connected (subtree) component  $\tset'$ in $\tset\setminus t$
such that, for any edge $ab\in E(P)$, $\tset'$ has a bag $t'$ that contains
$ab$, i.e., every edge $ab\in E_{t'}$ for some bag $t' \in V(\tset')$.
\end{lemma}

\begin{proof}
We provide a simple proof by contradiction. Assume that there are two
consecutive edges, $ab, bc \in E(P)$ that are in different connected
components of $\tset\setminus t$ (otherwise, all edges must be in the same
component). %
Since the set of bags containing $b$ must be connected in $\tset$ but is not
connected in $\tset\setminus \{t\}$, $b \in X_t$, we reach a contradiction.
\end{proof}

\begin{proof}[Proof of \Cref{lem:mixed:unify}]

We remark that the function $\tc^*$ shares some properties with the usual
definition of transitive closure, which are used throughout the proof:
\begin{observation}
The function $\tc^*$ satisfies the following properties:
\begin{itemize}
   \item $\tc^*_Z(\tc^*_Z(Y)) = \tc^*_Z(Y)$
   \item $\tc^*_{Z'}(Y') \subseteq \tc^*_Z(Y)$ if $Z'\subseteq Z$, $Y' \subseteq Y$
\end{itemize}
\end{observation}

\paragraph*{Equivalence for $Z_t$:} 

We prove that the two definitions for $Z$ are equivalent by induction on the
depth of the bag. At the root, we have that $W_{\rootn(\tset)} = W \cap X_{\rootn(\tset)}$, so the
equivalence holds. 

For the induction step, let $t \in V(\tset)$ be a bag other than the root. Then
\begin{align*}
   \left(Z_{p(t)} \cup W_t\right) \cap X_t &=         \left(Z_{p(t)} \cap X_t\right) \cup \left(W_t \cap X_t\right) \\ 
                                 &\subseteq \left(W \cap X_t\right) \cup \left(W \cap X_t\right) \\
                                 &=         W \cap X_t
\end{align*}
The second step follows from the induction hypothesis, as well as the definition of $W$.
We now prove the converse inclusion.
\begin{align*}
   W \cap X_t &=         \left(W \cap X_{p(t)} \cap X_t\right) \cup \left(W \cap \big(X_t \setminus X_{p(t)}\big)\right) \\
            &\subseteq \left(Z_{p(t)} \cap X_t\right) \cup W_t \\
            &=         \left(Z_{p(t)} \cup W_t\right) \cap X_t
\end{align*}
We use the induction hypothesis, as well as the fact that $W_t \subseteq X_t$.

\paragraph*{Equivalence for $\Gamma^*_t$:}

We prove the statement by induction on the height of a bag $t$. 
Since $E_{\tset_t} = \emptyset$, both definitions are equivalent for every leaf $t$.

Let $t$ be any bag. By the induction hypothesis, %
$\Gamma^*_{t'}, \Gamma^*_{t''} \subseteq \tc^*_W\left(Y \cap E_{\tset_t} \right)$. Therefore, 
\[
\restr{\tc^*_{Z_t}\left(\Gamma^*_{t'}\cup\Gamma^*_{t''}\cup Y_t\right)}_t \subseteq
\restr{\tc^*_W\left(Y \cap E_{\tset_t} \right)}_t
\]
Here, we use that $Z_t = W \cap X_t \subseteq W$.

To prove the converse inclusion, let $(u,v) \in \restr{\tc^*_W\left(Y \cap E_{\tset_t} \right)}_t$. 
By definition, there must be a path $p$ between $u$
and $v$ using internal vertices in $W$. Let $u=w_0, w_1, \ldots, w_\ell=v$ be all the
vertices of $p$ (in the correct order) that are also in $X_t$.

Each pair $(w_i, w_{i+1})$ is connected by a subpath of $p$. By \Cref{lemma:monochromatic},  $(w_i, w_{i+1})$ is either an edge in $Y_t$, or the
subpath is fully contained in $Y \cap E_{\tset_{t'}}$ or $Y \cap E_{\tset_{t''}}$, and uses internal vertices in $W$. In the
first case, $(w_i, w_{i+1}) \in Y_t$, while in the remaining
cases, $(w_i, w_{i+1})$ is in $\Gamma_{t'}$ or $\Gamma_{t''}$, by the induction
hypothesis.
We conclude that, since $w_i \in Z_t = W \cap X_t$, for $i \in [l-1]$, then $(u,v)$ is in 
\[\restr{\tc^*_{Z_t}\left(\Gamma^*_{t'}\cup\Gamma^*_{t''}\cup Y_t\right)}_t\]

\paragraph*{Equivalence for $\Delta^*_t$:}

We now prove that both definitions for $\Delta^*$ are equivalent, by using induction on the depth of the bags.
For the node $\rootn(\tset)$, we have $E_{\tset_{\rootn(\tset)}} = E(G)$, and thus the base case follows.

For the induction step, we remark that $\Delta^*_{p(t)}, \Gamma^*_t
\subseteq \tc^*_W\left(Y \right)$ by the induction hypothesis together with the
statement of the lemma for $\Gamma^*$. Therefore,
\[
\restr{\tc^*_{Z_t}\left(\Delta^*_{p(t)}\cup\Gamma^*_t\right)}_t \subseteq
\restr{\tc^*_W\left(Y \right)}_t
\]

For the reverse inclusion, we fix a pair $(u,v) \in \restr{\tc^*_W\left(Y \right)}_t$, 
and a path $p$ that connects $u$ to $v$ in $Y$ using internal vertices in $W$. 
Further, let $u=w_0, w_1, \ldots, w_\ell=v$ be all the
vertices of $p$ (in the correct order) that are also in $X_t$.

By \Cref{lemma:monochromatic}, $(w_i, w_{i+1})$ is either an edge in
$Y(b)$, or the subpath $p'$ of $p$ connecting $w_i$ and $w_{i+1}$ is contained
either in $Y \cap E_{\tset_{t'}}$, $Y \cap E_{\tset_{t''}}$ or $Y \cap (E \setminus E_{\tset_{t}})$.
For all but the last case, $(w_i, w_{i+1}) \in \Gamma^*_t$, by definition. 
In the remaining case, it must be that the vertices
of $p'$ are also contained in $X_{\tset \setminus \tset_b}$. Specifically,
because the bags containing a given vertex must form a connected component of
$\tset$,  $w_i, w_{i+1} \in X_{p(t)}$, which implies $(w_i, w_{i+1}) \in
\restr{\tc^*_W\left(Y \right)}_{p(t)}= \Delta^*_{p(t)}$.

In any case, since $w_i \in Z_t = W \cap X_t$, for $i \in [l-1]$, we conclude 
that every pair $(w_i, w_{i+1})$ and thus $(u,v)$, are contained in 
$\restr{\tc^*_{Z_t}\left(\Delta^*_{p(t)}\cup\Gamma^*_t\right)}_t$.
\end{proof}

\section{Details of the DP for \ecsndp} %
\label{sec:sec4proofs}

\subsection{Correctness}
\label{sub:ecew:correctness}

The following two lemmas imply the correctness of our DP. 

\begin{lemma}
\label[lemma]{lem:ecew:completeness}
Let $F \subseteq E(G)$ be a feasible solution. Then, for every bag $t$, there is a profile $\Psi_t$ for $t$ such that $(t, \Psi_t)$ is valid; for any bag $t$ with children $t'$, $t''$, then $\Psi \cons{Y_t} (\Psi', \Psi'')$, where $Y_t = F \cap E_t$.
Furthermore, $c[\rootn(\tset), \Psi_{\rootn(\tset)}] = c(F)$.
\end{lemma} 
\begin{proof}
We first observe that for each demand $(\gamma_i, k_i)$, $i\in [h]$ the solution $F$ can be partitioned into $F = \bigcup_{j\in [k]} F^i_j$ such that the partition $F^i_j$ contains a path connecting $\gamma_i$ to the root $r$. Note that, if $k_i < k$, there might not be any path in the partitions $F_{k_i+1}$, \ldots, $F_k$.  

Let $t$ be a bag. We define $\Psi_t$ as follows. 
We define the pair (omitting the script of $t$ for convenience) $(\vec{\Gamma}^{i}(t), \vec{\Delta}^i(t))$ for all $i\in [h]$ as follows. 
\begin{align*}
   \Gamma^i_j(t)  &= \restr{\tc(F^i_{j} \cap E(\tset_t))}_t \\
   \Delta^i_j(t)  &= \restr{\tc(F^i_{j})}_t
\end{align*}

We now define $\Psi_t = \{(\vec{\Gamma}^i(t),\vec{\Delta}^i(t)) : i\in [h] \}$.

We first show that any cell $(t, \Psi_t)$ is valid. For any leaf bag $t$, $E_t = \emptyset$ and hence $\Gamma^i_j(t) = \emptyset$, $i\in [h]$, $j\in [k]$. For the bag $\rootn(\tset)$, $E(\tset_t) = E$ and hence $\Gamma^i_j(t) = \Delta^i_j(t)$, $i\in [h]$, $j\in [k]$. Finally, since $\pset^i_j$ connects $r$ to $\gamma_i$ for all $i\in [h]$, $j\in [k_i]$, then $(r, \gamma_i)\in \Delta^i_j(t_{\gamma_i})$. We conclude that any cell $(t, \Psi_t)$ defined as above is marked valid. 

Finally we define, for any bag $t$, a set of edges $Y(t) \subseteq E_t$ along with a suitable partition of $Y(t)$. Let $Y(t) = F \cap E_t$ and $Y^i_j(t) = F^i_j\cap E_t$. The subsets $Y^i_j(t)$ indeed form a partition owing to the disjointness of the sets $F^i_j$, for all $j\in [k_i]$ and a fixed $i \in [h]$. 

It is clear from the above definitions of $\Gamma^i_j$ and $\Delta^i_j$ that they satisfy the global connectivity conditions for $Y^i_j$. Applying \Cref{lem:edge:unify}, the local conditions must also be satisfied for the bag $t$ along with its children $t'$, $t''$ and $Y(t)$. This gives us $\Psi \cons{Y(t)} (\Psi', \Psi'')$ and we are done.

Notice that the $c[\rootn(\tset), \Psi_{\rootn(\tset)}] = \sum_{t \in V(\tset)} c(Y(t)) = c(F)$.
\end{proof} 

We remark that the DP uses exactly one cell per bag in any valid solution.

\begin{lemma}
\label[lemma]{lem:ecew:soundness}
Let $\cset = \{(t, \Psi_t) : t\in V(\tset)\}$ be the set of cells selected by the dynamic programming solution. Then there exists a set of edges $F\subseteq E$ such that for each demand $(\gamma_i,k_i)$, $i \in [h]$, there exist $k_i$ edge-disjoint paths connecting $r$ to $\gamma_i$. Furthermore, $c(F) = c[\rootn(\tset), \Psi_{\rootn(\tset)}]$.
\end{lemma}

\begin{proof}
First we define the set $F$. Consider the cells in $\cset$. Since each cell $(t, \Psi_t)$ is picked by the DP, there must exist a set of edges $Y_t \subseteq E_t$ such that for some pair of children cells $\{(t', \Psi_{t'}), (t'', \Psi_{t''})\}\in \cset$, $\Psi_t \cons{Y_t}(\Psi_{t'}, \Psi_{t''})$. Define $F = \bigcup_{(t, \Psi_t)\in \cset} Y_t$. We prove that for any demand $(\gamma_i, k_i)$, $i\in [h]$, there exist edge-disjoint paths $\pset^i_j$, $j\in [k_i]$ in $F$ that connect $r$ to $\gamma_i$.   

For a demand vertex $\gamma_i$, $i\in [h]$, consider the bag $\ts=t_{\gamma_i}$ and the cell $(\ts, \Psi_{\ts}) \in \cset$. Since this cell is valid, there exists a connection profile $(\vec{\Gamma}, \vec{\Delta})\in \Psi_{\ts}$ such that $(r, \gamma_i) \in \Delta_j$ for all $j\in [k_i]$.

We now suitably define connection profiles $(\vec{\Gamma}^t, \vec{\Delta}^t)$ and sets of edges $Y_{t,j}$ for every other bag \mbox{$t\in V(\tset)$}, and prove that these elements satisfy the local connectivity property of \Cref{lem:edge:unify}. We explicitly describe the definition for the children bags $t'$ and $t''$ of $\ts$. The definition for every other bag in the subtree $\tset_t$ can be carried out in a similar recursive fashion.

By definition of consistent DP cells,  there exists  $(\vec{\Gamma}', \vec{\Delta}')\in \Psi_{t'}$, $(\vec{\Gamma}'', \vec{\Delta}'')\in \Psi_{t''}$ and a partition $Y_{\ts} = \bigcup_{j\in [k]} Y_{\ts,j}$ such that, for all $j \in [h]$, the triplet $(\Gamma_j, \Delta_j), (\Gamma'_j, \Delta'_j), (\Gamma''_j, \Delta''_j)$ satisfies the local connectivity conditions for $Y_{\ts,j}$. 
We take $(\vec{\Gamma}^{t'}, \vec{\Delta}^{t'}) = (\vec{\Gamma}', \vec{\Delta}')$, $(\vec{\Gamma}^{t''}, \vec{\Delta}^{t''}) = (\vec{\Gamma}'', \vec{\Delta}'')$.

A similar recursive definition works for all bags in $\tset\setminus \tset_t$, starting with $(\vec{\Gamma}, \vec{\Delta})$ and defining a suitable connection profile in $p(t)$.

We now apply the equivalence from \Cref{lem:edge:unify} to conclude that the global connectivity definition is satisfied by $(\Gamma^t_j, \Delta^t_j)$ and edge set $Y_{t,j}$ for any bag $t\in V(\tset)$. As a consequence, there exist paths connecting $r$ to $\gamma_i$ in $\bigcup_{t\in V(\tset)} Y_{t,j}$, for $j \in [k_i]$. Since the sets $Y_{t,j}$ are disjoint for $j\in [k_i]$, we obtain the required $k_i$ edge-disjoint paths in the sets $Y_j = \bigcup_{t \in V(\tset)} Y_{t,j}$.
\end{proof}

\subsection{Running Time Analysis} %
\label{sub:ecew:running_time}

To bound the running time of the DP algorithm, we start by proving a bound on the number
of cells in the DP table (\Cref{lem:ecew:dpsize}), and then show how this
implies the running time of the algorithm.

\begin{proof}[Proof of \Cref{lem:ecew:dpsize}] \quad

The following observations are used to prove the lemma:
\begin{observation} \quad
\label[observation]{obs:sizes}

\begin{enumerate}
   \item The number of possible subsets of edges between $w$ elements is $2^{w^2}$.
   \item The number of possible partitions of such a subset of edges is $k^{w^2}$.
   \item The number of possible equivalence relations in a set of $w$ elements is $w^w$.
\end{enumerate}
\end{observation}

Let $t \in \tset$. We know that $\Psi \subseteq \xset_t \times \xset_t$, whose size can be up to
$2^{kw^2}$. However, it is sufficient to consider equivalence relations
in the bag $t$, as we always take the transitive closure in definitions.
Therefore,  there are at most $w^w$ possibilities for each $\Gamma_j,
\Delta_j$ (\Cref{obs:sizes}) and at most $2^{w^{2wk}}$ possibilities
for $\Psi$.
\end{proof}

Note that the algorithm itself does one of the two possible options for each cell:
it either initializes itself, for which it needs to check every element of $\Psi$,
taking time $O(w^{2wk}kw)$; or it computes the value based on children cells, for which it
enumerates all sets $\Psi', \Psi''$ and checks them for consistency. The
number of such sets to check is again $2^{w^{2wk}}$, and checking each triple
of sets takes time polynomial in $w^{wk}$ and $k^{w^2}$.
In sum, the algorithm takes time $O(n\,\exp(w^{O(wk)}))$.

\section{Details on the Algorithms for \RestrictedGroupSNDP}

In this section, we present \Cref{lem:rgsndp:completeness} and
\Cref{lem:rgsndp:soundness}, which prove the correctness of the reduction
presented in \Cref{sec:rgsndp} and complete the proof of
\Cref{thm:rgsndp:tree_instance}.

\begin{lemma}
\label[lemma]{lem:rgsndp:completeness}
For every solution $F \subseteq E(G)$, there is a valid tree $\widetilde T \subseteq \widetilde \tset$ 
with the same cost and that connects the same groups that $F$ 
$k_i$-connects, that is, if $F$ contains $k_i$ edge-disjoint paths to a vertex
$\gamma_i \in S_i$, then $\widetilde T$ connects the root to some node $\ttil \in \widetilde
S_i$.
\end{lemma}
\begin{proof}
Let $\{\gamma_i \in S_i\}_{i \in [h']}$, $h' \leq h$, be the group vertices that are
$k_i$-connected by $F$  (w.l.o.g. $F$ the first $h'$ groups). %
For every $i \in [h']$, we partition the solution $F$ into $k_i$ subsets that
connect $r$ to $\gamma_i$, i.~e.~we partition $F$ into $\{F_{ij}\}_{j \in
[k_i]}$, such that each $F_{ij}$ contains a $r$-$\gamma_i$-path. We also define
$\Psi_t$ and $(\vec{\Gamma}^i(t), \vec{\Delta}^i(t)) \in \Psi_t$ for all $t
\in V(\tset)$ as is done in the proof of \Cref{lem:ecew:completeness}.

By the proof of \Cref{lem:ecew:completeness}, $(\Psi_{t'}, \Psi_{t''}) \cons{F
\cap E_t} \Psi_t$. Therefore, we can build a tree $\widetilde T$ by picking
one copy of each of the nodes $\ttil[t, \Psi_t]$ and connecting nodes
$\ttil_c[t, \Psi_t, \Psi_{t'}, \Psi_{t''}, F \cap E_t]$, such that $\widetilde
T$ is connected. %

Furthermore, for the topmost bag $t_{\gamma_i}$ containing $\gamma_i$,  we have that %
$\ttil[t_{\gamma_i}, \Psi_{t_{\gamma_i}}] \in \widetilde S_i$, since $F_{ij}$ connects $r$ and $\gamma_i$, which implies that $(r,\gamma_i)
\in \Delta^i_j(t_{\gamma_i})$, for $j \in [k_i]$. Therefore, if $F$
$k_i$-connects a group $S_i$, $\widetilde T$ contains a node of $\widetilde S_i$.

The proof that $F$ and $\widetilde T$ have the same cost follows from the proof
of \Cref{lem:ecew:completeness}.
\end{proof}

\begin{lemma}
\label[lemma]{lem:rgsndp:soundness}
For every valid tree $\widetilde T \subseteq \widetilde \tset$, there is a solution $F
\subseteq E(G)$ with the same cost and which $k_i$-connects a group $S_i$ if 
$\widetilde T$ connects $\widetilde S_i$, that is, %
if $\widetilde T$ connects the root to some node $\ttil \in \widetilde S_i$, then $F$
contains $k_i$ edge-disjoint paths to a vertex $\gamma_i \in S_i$.
\end{lemma}
\begin{proof}
Let $\{\widetilde S_i\}_{i \in [h']}$, $h' \leq h$ be the groups connected by
$\widetilde T$ (w.l.o.g.). %
By the definition of $\widetilde S_i$, there must be some $\gamma_i \in S_i$ and
node $\ttil[t_{\gamma_i}, \Psi_{t_{\gamma_i}}] \in \widetilde S_i \cap V(\widetilde T)$ for
every $i \in [h']$.

By the definition of valid tree, $\widetilde T$ contains exactly one node
$\ttil[t, \Psi_t]$ for each $t \in V(\tset)$, as well as exactly one node
$\ttil_c[t, \Psi_t, \Psi_{t'}, \Psi_{t''}, Y_t]$ for each non-leaf node
$t \in V(\tset)$ such that  
\[\Psi_t \cons{Y_t}(\Psi_{t'}, \Psi_{t''}).\]

The above conditions suffice to apply the proof of 
\Cref{lem:ecew:soundness} (with demands $\{(\gamma_i, k_i)\}_{i \in [h']}$);
thus, we can obtain a solution $F$ such that for every $i \in [h']$, $F$
contains $k_i$ edge-disjoint paths from $r$ to $\gamma_i$. Furthermore, $F$ has
the same cost as $\widetilde T$.
\end{proof}

\acksandrefsFull{}

\section{Algorithms for \vcsndp} 
\label{sec:vc-sndp} 

In this section, we focus on vertex connectivity problems,
specifically, \RootedVCSNDP with multiple roots, 
where there is a subset $R \subseteq V$ such that the connectivity requirement $k(u,v) > 0$ only if $u \in R$. 

We argue that any instance of \SubsetkVC can be transformed into that of \RootedVCSNDP with $k$ roots as follows: 
Let $T$ be a terminal set for $k$-subset vertex connectivity. 
Let $R \subseteq T$ be arbitrary subset of terminals with $|R| = k$. 
We specify the connectivity $k(r,v) = k$ for all $r \in R$ and $v \in T$. It is easy to verify that this instance is equivalent to the original instance.  
In the subsequent discussion, we focus on \RootedVCSNDP with at most $k$ roots. 
We start by adding, for each $r \in R$, vertex $r$ into every bag $X_t$. 
This increases the width of the decomposition by an additive factor of at most $k$. 

The organization in this section is as follows. 
We start by extending the framework of \Cref{sec:new-concept} to vertex-connectivity in \Cref{sub:vcnw:extension}.
In \Cref{sub:vcnw:profiles}, we explain how the DP table is setup and the
main ideas for the vertex-connectivity version of the problem. %
In \Cref{sub:vcnw:dp}, we present the DP algorithm, and how to compute the value
of each cell. %
Finally, in \Cref{sub:vcnw:correctness} we show that the algorithm
finds an optimal solution to the \vcsndp{} problem. %
Throughout this section, we will use the notation defined in %
\Cref{sec:ec-sndp}, when applicable.

\section{Global \texorpdfstring{$\Leftrightarrow$}{^^e2^^87^^94} Local Checking for Vertex Connectivity}
\label{sub:vcnw:extension}

When dealing with vertex connectivity, we want to find paths that have
disjoint internal vertices, but which all share the same source and sink. %
Therefore, we must introduce some changes to our connection sets and to 
the way they are defined.

The most important change is the use of a function $\tc^*$, which is a
modified version of transitive closure, introduced in
\Cref{appendix:proofs:unify}. We recall its definition here:
\[
\tc^*_Z(S) = \bigl\{(u,v) \bigm\vert \exists w_1, \ldots, w_\ell \in Z, \forall i \in [\ell-1], (u,w_1), (w_\ell, v), (w_i, w_{i+1}) \in S \big\}
\]

Our previous notion of pair of sets $(\Gamma, \Delta)$ is now extended
with a subset $Z \subseteq X_t$, which represents the vertices that can be
used as internal nodes in a bag. We denote the new triple $(Z, \Gamma^*,
\Delta^*)$. %
Let $W_t \subseteq X_t \setminus X_{p(t)}$ for every $t \in V(\tset)$, $W = \bigcup_{t \in
V(\tset)} W_t$, and a triple $(Z, \Gamma^*_t, \Delta^*_t)$ for every $t \in
V(\tset)$.

Given $W$, we say that the triples $(Z_t,\Gamma^*_t, \Delta^*_t)$ satisfy the \emph{local} (resp. \emph{global})
\emph{connectivity definition} if, for every bag $t \in V(\tset)$,

\noindent
\begin{minipage}[t]{.5\textwidth}
     \vspace{-2em}
     \begin{align*}
          \text{\bf Local}& \\
          Z_t      &:=
          \begin{cases}
               W_t                                                   & \text{if } t = \rootn(\tset) \\
               \left(Z_{p(t)} \cup W_t\right) \cap V_t         ~~~~~~& \text{otherwise}
          \end{cases} \\
          \Gamma_t^* &:= 
          \begin{cases}
               \emptyset                                                        & \text{$t$ is a leaf bag} \\
               \restr{\tc^*_{Z_t}(\Gamma^*_{t'}\cup\Gamma^*_{t''}\cup E_t)}_t   & \text{otherwise} 
          \end{cases} \\
          \Delta_t^* &:= 
          \begin{cases}
               \Gamma^*_t                                                   & \text{if } t = \rootn(\tset) \\
               \restr{\tc^*_{Z_t}(\Delta^*_{p(t)}\cup\Gamma^*_t)}_t   ~~~~~~& \text{otherwise} 
          \end{cases} 
     \end{align*}
     \vspace{.3em}
\end{minipage}%
\begin{minipage}[t]{.5\textwidth}
     \vspace{-2em}
     \begin{align*}
          \text{\bf Global}& \\
               Z_t      &:= W \cap X_t \\
               \Gamma^*_t &:= \restr{\tc^*_W\left(E_{\tset_t} \right)}_t \\
               \Delta^*_t &:= \restr{\tc^*_W\left(E(G)\right)}_t
     \end{align*}
     \vspace{.3em}
\end{minipage}

As in \Cref{sec:new-concept}, we prove that the the local and global
connectivity definitions are equivalent (\Cref{lem:vertex:unify}).
\begin{lemma}
\label[lemma]{lem:vertex:unify}
Let $W_t \subseteq X_t \setminus X_{p(t)}$ be a subset of vertices and $(Z,
\Gamma^*_t, \Delta^*_t)$ a triple for every $t \in V(\tset)$. %
Then, the triples $(Z_t, \Gamma^*_t, \Delta^*_t)$ satisfy the local
connectivity definition iff they satisfy the global one.
\end{lemma}

The proof of the lemma follows from \Cref{lem:mixed:unify}, by setting $Y_t = E_t$.

The notions of local and global connectivity in \Cref{lem:vertex:unify} are
also used to design algorithms for edge-connectivity with node-costs. %
However, we do not need the modified definition of transitive closure
$\tc^*(\cdot)$ and use the $\tc(\cdot)$ operator instead. We omit the full definition to avoid repetition. %
We remark that, even though we can use $\tc^*(\cdot)$, even for edge-connectivity, we can achieve a better runtime using $\tc(\cdot)$, as is shown in \Cref{sec:ec-sndp}.

\subsection{Profiles}
\label{sub:vcnw:profiles}

As before, let $\xset_t$ be the set of all connection profiles for $t$. By
analogy, we say a {\em vertex profile} for $t$ is a $k$-tuple 
$\vec{Z}=(Z_1,Z_2,\ldots, Z_k)$ is such that $Z_i \subseteq X_t$. 
Additionally, let $\vset_t$ denote the set of all vertex
profiles for $t$.

The profile $\Psi$ of node $t$ is now a collection of triples %
$(\vec{Z}, \vec{\Gamma}^*, \vec{\Delta}^*)$, where $\vec{Z}$ is a vertex profile
and $\vec{\Gamma}^*$, $\vec{\Delta}^*$ are connection profiles,
i.e., $\Psi \subseteq \vset_t \times \xset_t \times \xset_t$.

A partial solution $H \subseteq V(G_t)$ is said to be {\em consistent} with
profile $\Psi$ for $t$ if, for all $(\vec{Z}, \vec{\Gamma}^*, \vec{\Delta}^*) \in
\Psi$,
\begin{itemize} 
\item For each $(u,v) \in \Gamma^*_i$ and $(a,b) \in \Gamma^*_j$ for $i\neq j$,
there are paths $P_{uv}, P_{ab} \subseteq H$ connecting the respective vertices in
$G_t$ such that the internal vertices of $P_{uv}$ and $P_{ab}$ are disjoint. %
All the internal vertices in $P_{uv}$, $P_{ab}$ which are in $X_t$ are
contained in $Z_i$, $Z_j$, respectively.

\item There is a global solution $H' \supseteq H$ such that, for each $(u,v)
\in \Delta_i$ and $(a,b) \in \Delta_j$ for $i\neq j$, there are paths $Q_{uv},
Q_{ab} \subseteq H'$ connecting the respective vertices in $G$ such that the
internal vertices of $Q_{uv}$ and $Q_{ab}$ are disjoint.
All the internal vertices in $Q_{uv}$, $Q_{ab}$ which are in $X_t$ are
contained in $Z_i$, $Z_j$, respectively.
\end{itemize}

In other words, a solution consistent with a profile must ``implement'' all
connectivity requirements by $\vec{\Gamma}$ and must be extensible to satisfy
$\vec{\Delta}$. Furthermore, the vertices used by the solution
in $X_t$ are given by $\vec{Z}$.

The DP table has a cell $c[t,\Psi]$ for each bag $t \in V(\tset)$ and each profile $\Psi$ for $t$. 
This cell tentatively stores the optimal cost of a solution consistent with the profile.

\subsection{The DP}  
\label{sub:vcnw:dp}

\shortparagraph{Valid cells:}
We mark all the cells that do not correspond to valid solutions of the problem
as invalid (for example, by setting $c[t, \Psi] = \infty$ for invalid cells). %
In particular, the following cells are invalid: 
\begin{itemize} 
\item Any leaf that has any connectivity requirements is invalid. That is, we set $c[t, \Psi] = 0$ if 
$\Psi \subseteq \vset_t \times\{(\emptyset, \ldots, \emptyset)\} \times \xset_t$; otherwise, $c[t, \Psi] = \infty$.  

\item Any cell that cannot be extended into a feasible solution is invalid. 
Let $t_v$ be the topmost bag that contains a terminal $v \in T$. %
The cell is valid if, for every $r \in R$, there is a triple
$(\vec{Z},\vec{\Gamma}, \vec{\Delta}) \in \Psi$ such that $(r, v) \in
\Delta_j$ for all $j \leq k(r,v)$ and $r,v \in \bigcup_{j \in [k]} Z_i$. %

Otherwise, either $r$ or $v$ are not in the solution, or they are not
connected by $k(r,v)$ vertex-disjoint paths, and hence the cell is invalid.

\item The root bag $\rootn(\tset)$ together with profile $\Psi$ is an invalid
cell if there is a triple $(\vec{Z}, \vec{\Gamma}, \vec{\Delta})\in \Psi$ such
that $\Gamma_j \neq \Delta_j$ for some $j$. 
\end{itemize}

\shortparagraph{DP Computation:} For all other cells, we compute their values from the values of the children. %
Let $t$ be a bag with left-child $t'$ and right-child $t''$. %
Let $\Psi, \Psi', \Psi''$ be their profiles, respectively, and $W \subseteq X_t \setminus X_{p(t)}$.

We say that $\Psi$ is consistent with $(\Psi', \Psi'')$ via $W$ (abbreviated
by $\Psi \cons{W} (\Psi', \Psi'')$) if the following
conditions are satisfied.
For each triple $(\vec{Z}, \vec{\Gamma}, \vec{\Delta}) \in \Psi$, there are
$(\vec{Z}', \vec{\Gamma}' , \vec{\Delta}') \in \Psi'$ and \mbox{$(\vec{Z}'',
\vec{\Gamma}'' , \vec{\Delta}'') \in \Psi''$}, together with a partition of $W$
into $W_1 \cup W_2 \cup \ldots \cup W_k$ such that, for all $j \in [k]$,
\begin{itemize}
\item $\Gamma_j = \restr{\tc^*_{Z_t}(\Gamma'_j \cup \Gamma''_j \cup E_t)}_t$
\item $\Delta'_j = \restr{\tc^*_{Z_{t'}}(\Delta_j \cup \Gamma'_j)}_{t'}$
\item $\Delta''_j = \restr{\tc^*_{Z_{t''}}(\Delta_j \cup \Gamma''_j)}_{t''}$
\item $Z_j \cap (X_t \setminus X_{p(t)}) = W_j$
\item $Z'_j \cap X_t = Z_j \cap X_{t'}$
\item $Z''_j \cap X_t = Z_j \cap X_{t''}$
\end{itemize}

Similarly, for any triple $(\vec{Z}', \vec{\Gamma}', \vec{\Delta}') \in \Psi'$ %
(resp., $(\vec{Z}'', \vec{\Gamma}'', \vec{\Delta}'') \in \Psi''$), %
there are $(\vec{Z}'', \vec{\Gamma}'', \vec{\Delta}'') \in \Psi''$ %
(resp., $(\vec{Z}', \vec{\Gamma}', \vec{\Delta}') \in \Psi'$), %
plus $(\vec{Z}, \vec{\Gamma}, \vec{\Delta}) \in \Psi$ and some partition of $W$  satisfying similar conditions as above.

We remark that the condition $Z_t = (Z_{p(t)} \cup W_t) \cap X_t$ in the local
connectivity definition of \Cref{lem:vertex:unify}, can be equivalently
written as $Z_t \cap X_{p(t)} = Z_{p(t)} \cap X_t$ and $Z_t \cap (X_t
\setminus X_{p(t)}) = W_t$ (since $W_t$ and $Z_{p(t)}$ are always disjoint). %
We choose to use the second formulation for convenience.

Then the value of $c[t, \Psi]$ can be defined recursively among valid cells: 
\[c[t, \Psi] = \min_{\Psi \cons{W} (\Psi', \Psi'')} \left( c[t', \Psi'] + c[t'', \Psi''] + c(W) \right) \]

\shortparagraph{Solution:}
The solution can be computed by 
\[
\min \left\{ c[\rootn(\tset),\Psi] \mid (\forall (\vec{Z}, \vec{\Gamma}, \vec{\Delta}) \in \Psi)(\forall j \in [k]) \Gamma_j = \Delta_j \right\}
\] 

\subsection{Correctness of the DP}
\label{sub:vcnw:correctness}
The following two lemmas imply correctness. 

\begin{lemma} 
\label[lemma]{lem:vcnw:completeness}
Let $H \subseteq V(G)$ be a feasible solution. %
Then, for every bag $t$, there is a profile $\Psi_t$ for $t$ such that %
$(t, \Psi_t)$ is valid; for any bag $t$ with children $t'$, $t''$, %
then $\Psi \cons{W_t} (\Psi', \Psi'')$, where $W_t = H \cap (X_t \setminus X_{p(t)})$. %
Furthermore, $c[\rootn(\tset), \Psi_{\rootn(\tset)}] = c(H)$.
\end{lemma} 
\begin{proof}

We fix $r \in R$, $v \in T$ with $k(r,v) > 0$. Observe that $r, v \in H$, and
the solution $H$ can be partitioned into $H = \{H_j\}_{j \in [k]}$ such that
for $j \leq k(r,v)$, there is a path $\pset_j$ between $r$ and $v$ whose
internal vertices are in $H_j$. %
Note that, if $k(r,v) < k$, then
$H_j$ is irrelevant, for $j > k(r,v)$, and might be empty.

For any bag $t$, we define $\Psi_t$ as follows. %
The triple $(\vec{Z}(t), \vec{\Gamma}^{*}(t), \vec{\Delta}^{*}(t))$ is defined as: %
\begin{align*}
   Z_j(t)          &:= H_j \cap X_t \\
   {\Gamma^*}_j(t) &:= \restr{\tc^*_{H_j}\left(E_{\tset_t} \right)}_t \\
   {\Delta^*}_j(t) &:= \restr{\tc^*_{H_j}\left(E(G)\right)}_t
\end{align*}

We add to $\Psi_t$ the triples $(\vec{Z}(t), \vec{\Gamma}^{*}(t), \vec{\Gamma}^{*}(t))$ for all $r \in R$, $v \in T$.

We first show that any cell $(t, \Psi_t)$ is valid. %
We again fix $r \in R$, $v \in T$ and take the corresponding triples $(\vec{Z}(t), \vec{\Gamma}^{*}(t),
\vec{\Delta}^{*}(t)) \in \Psi_t$ and partition $\{H_j\}_{j \in [k]}.$ %

For every leaf bag $t$, $E_t = \emptyset$ and hence ${\Gamma^*}_j(t) = \emptyset$. %
For the root bag, $t=\rootn(\tset)$, $E(\tset_t) = E$ and hence ${\Gamma^*}_j(t) = {\Delta^*}_j(t)$. %
Finally, for $j \in [k(r,v)$, $\pset_j$ connects $r$ to $\gamma_i$ using
internal vertices in $H_j$, so $(r, v) \in {\Delta^*}_j(t_{v})$. %
Since, $r, v \in \bigcup_{j \in [k]} Z_j(t_{v})$, we conclude that
the cells $(t, \Psi_t)$ defined above are marked valid.

For each bag $t \in V(\tset)$, we additionally define $W_j(t) = H_j \cap (X_t
\setminus X_{p(t)})$. Observe that the $W_j(t), j = 1,2,\ldots,k_i$ naturally induce a partition on $W_t$. It is clear from the definitions that $Z_j(t)$,
${\Gamma^*}_j(t)$, ${\Delta^*}_j(t)$ satisfy the global connectivity
definitions for vertex sets $W_j(t)$. Therefore, we can apply \Cref{lem:vertex:unify}, 
for every $j \in [k]$ to conclude that
\[
\Psi_t \cons{W_t} (\Psi_{t'}, \Psi_{t''})
\]
for every bag $t \in V(\tset)$ with children $t'$, $t''$, where $W_t = H \cap (X_t \setminus X_{p(t)})$.

Notice that the $c[\rootn(\tset), \Psi_{\rootn(\tset)}] = \sum_{t \in V(\tset)} c(W_t) = c(H)$.
\end{proof} 

We remark that every DP solution uses exactly one cell per bag $t \in V(\tset)$.

\begin{lemma}
\label[lemma]{lem:vcnw:soundness}
Let $\{(t, \Psi_t)\}_{t\in V(\tset)}$ be the cells selected in a dynamic programming solution. %
There exists a solution $H \subseteq V$ with cost $c(H) = c[\rootn(\tset), \Psi_{\rootn(\tset)}]$ %
such that for every demand $(\gamma_i, k_i)$, $i \in [h]$,
$H$ contains $k_i$ vertex-disjoint paths from $r$ to $\gamma_i$.
\end{lemma}

\begin{proof}
Let us start by defining the solution $H$. For every non-leaf $t \in V(\tset)$ with children $t'$, $t''$,
there must be some set $W_t$ such that
\[
\Psi_t \cons{W_t} (\Psi_{t'}, \Psi_{t''})
\]
We take $H = \bigcup_{t \in V(\tset)} W_t$ as our solution.

Now, fix $r \in R$, $v \in V$. %
We will prove that $H$ contains $k(r,v)$ vertex-disjoint paths from $r$ to $v$. %
We now define, for each bag $t\in V(\tset)$, the sets $(\vec{Z}(t), \vec{\Gamma}^*(t), \vec{\Delta}^*(t)) \in
\Psi_t$, as well as partition $\{W_j(t)\}_{j \in [k]}$ of the vertices in
$W_t$, as is done in the proof of \Cref{lem:ecew:soundness}.

Now, since for every $j \in [k]$, $Z_j(t)$, $\Gamma^*_j(t)$, $\Delta^*_j(t)$
satisfy the local connectivity constraints, we can apply \Cref{lem:vertex:unify}, which implies that $Z_j(t)$, $\Gamma^*_j(t)$,
$\Delta^*_j(t)$ satisfy the global connectivity constraints as well. %

Since all cells in $(t, \Psi_t)$ are valid, we know that $(r, v) \in
\Delta^*_j(t)$ for $j \in [k(r,v)]$ and $r, v \in \bigcup_{j \in [k]}
Z_j(t)$. %
Therefore, $H_j = \bigcup_{t \in V(\tset)} H_j(t)$ contains the internal nodes of
a path between $r$ and $v$, and $r, v \in H$. We conclude that
$H$ contains $k(r,v)$ vertex-disjoint paths between $r$ and $v$.
\end{proof} 

\subsection{Running Time Analysis} %
\label{sub:vcnw:running_time}

The running time for the presented algorithm follows closely that for the
edge-connectivity version, with the exception that $\Gamma^*$ and $\Delta^*$
are, in general, no longer equivalence relations, and there are now $2^{(w+k)^2}$
possibilities for each such set.

We conclude that the running time of the algorithm is $O(n\cdot \exp(\exp(O((w+k)^2k))))$.
We also remark that a very similar idea for the DP works for the edge-connectivity version of the problem with node-costs. However, we can use the definition of $\tc$ instead of $\tc^*$ in order to define the DP cells and hence the runtime for that case is $O(n\cdot \exp(\poly(w^{wk}))$.

\section{Hardness of Restricted Group SNDP on General Graphs}
\label{app:hardness}

In this section, we sketch the proof of the approximation hardness
for the vertex-weighted \RestrictedGroupSNDP on general graphs.
The reduction is from {\em Min-$k$-CSP}.

In {\em Max-$k$-CSP}, we are given a set of $n$ variables $x_1,\ldots,x_n$
over the domain $[N]$ and a set of $m$ constraints (which are functions) $C_1,\ldots,C_{h}$
such that each $C_j$ depends on (exactly) $k$ variables.
The goal in the {\sf Max-$k$-CSP} is to find an assignment to variables 
that maximizes the number of constraints satisfied.
In the minimization version, say {\em Min-$k$-CSP}, 
we are allowed to assign multiple values (or {\em labels}) to each variable 
to guarantee that each constraint can be satisfied 
by some of the assignments while minimizing the total number of labels.
To be specific the assignment is a function $\sigma:\{x_1,\ldots,x_n\}\rightarrow2^{[N]}$, and we wish to find an assignment such that,
for every constraint $C_j$ depending on variables $x_{i_1},\ldots,x_{i_k}$, 
there exist labels $a_1\in \sigma({x_{i_1}}), \ldots, a_k\in \sigma({x_{i_k}})$
such that $C_j(a_1,\ldots,a_n)$ evaluates to {\em true} for all $j\in[m]$.
It is know that if the hardness of {\sf Max-$k$-CSP} is $\alpha(n)$, then the hardness of
the {\sf Min-$k$-CSP} is $\Omega(\alpha(n)^{1/k})$ (see, e.g., \cite{Laekhanukit14}).
The special case of $k=2$ is called the {\em Label-Cover} problem.

We reduce from {\sf Min-$k$-CSP} to \RestrictedGroupSNDP by representing
the assignments of variables by positive weight vertices 
and representing accepting configurations of each constraints by a group of vertices.
To be precise, we first construct a graph $G=(V,E)$ such that
\begin{align*}
V &= \{r\} \cup \{(x_i,a): i\in[n], a\in [N]\} \cup \\
  &~~~~ \{(j,a_1,\ldots,a_k): j\in [h], a_1,\ldots,a_k\in[N] \mbox{ and }
         C_j(a_1,\ldots,a_k)=\mathrm{true}\}\\
E &= \{r(x_i,a): i\in[n], a\in [N]\} \cup \\
  &~~~~ \{(x_i,a_i)(j,a_1,\ldots,a_k): (x_i,a_i),(j,a_1,\ldots,a_k)\in V \mbox{ and } \\
  &~~~~~~ \mbox{$C_j$ depends on $x_i$ at the $\ell$-th argument, and $a_\ell=a$}\}\\
\end{align*}
We set the cost of each vertex of the form $(x_i,a)$ to be one
and set cost of other vertices to be zero.
We define $r$ to be the root vertex and, for each $j\in [h]$,
we define a group $S_j=\{(j,a_1,\ldots,a_k)\in V\}$.
We set the connectivity demands $k_j=k$ for all $j\in[h]$.
This finishes the construction.

Observe that every vertex that belongs to some group $S_j$ has degree exactly $k$.
Thus, the only way to have $k$-edge disjoint paths from some vertex in $S_j$ to
the root $r$ is to choose all of its neighbors in 
$\{(x_i,a): \mbox{$C_j$ depends on $x_i$}, a\in [N]\}$.
It is not hard to see that vertices in the latter set corresponding to 
an assignment to variables in the Min-$k$-CSP instance
and that some vertex in $S_j$ has $k$ neighbors if and only if 
the chosen vertices correspond to an assignment that satisfies the constraint $C_j$.
Moreover, ones may verify that a vertex $v\in S_j$ has $k$ neighbors in an induced subgraph $H\subseteq G$
if and only if $v$ has $k$ edge-disjoint paths connecting to the root $r$.
Consequently, this gives a reduction from {\sf Min-$k$-CSP} to 
the vertex-weighted \RestrictedGroupSNDP.
The approximation hardness then follows from the hardness of {\sf Min-$k$-CSP},
which is $2^{\log^{1-\epsilon}n}$ under $\NP\not\subseteq\DTIME(n^{\polylog{n}})$,
and $n^{\delta/k}$ for some constant $0<\delta<1$ under the Sliding Scale Conjecture \cite{BellareGLR93}.
\end{document}